\documentclass{article}

 \usepackage[dblblindworkshop, final]{neurips_2026}
\workshoptitle{TAE (Trust-AI-Eval): Can We Trust AI Evaluation?}

\setcitestyle{round}

\usepackage[utf8]{inputenc} 
\usepackage[T1]{fontenc}    
\usepackage{hyperref}       
\hypersetup{hidelinks}
\usepackage{url}            
\usepackage{booktabs}       
\usepackage{amsfonts}       
\usepackage{nicefrac}       
\usepackage{microtype}      
\usepackage{xcolor}         
\usepackage{multirow}
\usepackage{graphicx}
\usepackage{comment}

\usepackage{amsmath}
\usepackage{amssymb}
\usepackage{bbm}
\usepackage{array}
\usepackage{tabularx}
\newcommand{\R}{\mathbb{R}}
\newcommand{\E}{\mathbb{E}}

\usepackage{amsthm}
\theoremstyle{plain}

\newtheorem{prop}{Proposition}[section]

\theoremstyle{definition}

\theoremstyle{remark}
\newtheorem*{remark}{Remark}

\title{{\Large Can We Trust Evaluations of AI Weather Forecasts for Extremes? How Evaluation Protocols Shape Model Rankings}}

\author{
    Tobias Biegert \\
    Institute of Statistics \\
    Karlsruhe Institute of Technology \\
    76185 Karlsruhe, Germany \\
    \texttt{tobias.biegert@kit.edu} \\
    \And
    Sam Allen \\
    Institute of Statistics \\
    Karlsruhe Institute of Technology \\
    76185 Karlsruhe, Germany \\
    \texttt{sam.allen@kit.edu} \\
    \And
    Zhongwei Zhang \\
    Institute of Statistics \\
    Karlsruhe Institute of Technology \\
    76185 Karlsruhe, Germany \\
    \texttt{zhongwei.zhang@kit.edu} \\
    \And
    Sebastian Lerch \\
    Dept. of Mathematics and Computer Science \\
    Marburg University \\
    35032 Marburg, Germany \\
    \texttt{sebastian.lerch@uni-marburg.de} \\
}

\begin{document}

\maketitle

\begin{abstract}
  AI weather models perform competitively with traditional physics-based models at a fraction of the computational cost, and they are thus already used operationally. However, studies reach conflicting conclusions about their ability to forecast extreme weather events, and comparisons across studies are complicated by differing evaluation protocols. We discuss the challenges that arise when comparing AI and physics-based weather forecasts for extreme events. Within a common experimental design, we find that model rankings are highly sensitive to the choice of evaluation metrics: the considered physics-based model tends to perform best when evaluation is restricted to observed extreme events, whereas AI models generally perform better when evaluated using weighted scoring functions. A contingency table analysis shows that AI models underpredict the frequency of extremes, so that their scores are generally dominated by misses, while the physics-based model forecasts extremes more frequently and its scores are dominated by false alarms. Different evaluation metrics penalise these errors differently. Further score decompositions confirm that the score differences arise primarily from miscalibration rather than discrimination ability. This sheds light on the contrasting results in the literature, and highlights the importance of the evaluation metrics used to evaluate AI models.
\end{abstract}

\section{Introduction}
\label{sec:intro}

Recent large-scale data-driven (or AI) weather prediction models have demonstrated competitive performance with traditional physics-based numerical weather prediction (NWP) models at a fraction of the computational cost \citep{Bi2023pangu,Lam2023graphcast, chen2023fuxi,lang2024aifs,Bodnar2025aurora}. 
The development of data-driven weather models has accelerated substantially in the past few years, and AI models are already being employed at operational weather centres.

The value of weather forecasts lies in their ability to aid decision making. Since extreme weather events typically lead to the largest socio-economic impacts, accurate and reliable forecasts for these events are particularly valuable. Several studies have therefore examined how well AI models can predict extreme weather events: targeted evaluations report competitive or superior performance when predicting tropical-cyclone tracks \citep{Bi2023pangu,Lam2023graphcast,Bodnar2025aurora,DeMaria2025} and in some classification-based evaluations of temperature extremes \citep{Lam2023graphcast,BenBouallegue2024}, results for atmospheric rivers are more metric-dependent \citep{Lam2023graphcast,Davis2026}, while several studies identify persistent underestimation of tropical- and midlatitude-cyclone intensity and peak winds \citep{BenBouallegue2024,DeMaria2025,Dacre2026}. Case studies likewise yield event-dependent results \citep{CharltonPerez2024,Pasche2025}, motivating broader comparisons over common samples of extreme events.

However, independent evaluation studies that systematically compare AI and physics-based models when forecasting extreme events have produced seemingly conflicting results: \citet{Zhang2026} conclude that AI models are less capable of predicting record-breaking heat, cold, and wind extremes than physics-based models, whereas \citet{OlivettiMessori2024} find that AI weather models marginally outperform physics-based models when predicting the most extreme $5\%$ heat and wind events, and \citet{Biegert2026twpcrps} find that AI models generally issue more informative predictions. 

These contrasting results may partly reflect differences in how the models are evaluated. Generally, AI model performance is not characterized by a single intrinsic quantity, but depends on the data, task, evaluation metric, and comparison protocol. Consequently, benchmark design can affect both reported performance and model rankings \citep{Flach2019,liao2021are}. 
Similar design choices arise when evaluating AI weather prediction models, which are often compared using different datasets, forecast models, definitions of extreme events, and evaluation metrics. While existing benchmark datasets \citep[e.g., WeatherBench 2,][]{Rasp2024weatherbench2} facilitate standardized comparisons, important choices remain regarding how model performance should be assessed, with additional theoretical challenges arising for extreme events \citep{Lerch2017dilemma,BrehmerStrokorb2019}.

In this paper, we examine the challenges involved in fairly comparing AI and physics-based models when predicting extreme weather events. Within a common experimental design, we compare deterministic (i.e. point-valued) forecasts from three state-of-the-art AI weather models to a physics-based model employed operationally at the European Centre for Medium-Range Weather Forecasting (ECMWF). The models are compared when predicting record-breaking heat, cold, and wind events. Predictive performance is assessed using multiple evaluation metrics. To understand differences between metrics, we examine contributions to the scores from hits, misses, and false alarms, and we use score decompositions to quantify how informative and how miscalibrated the forecasts are. 

\section{Experimental design}
\label{sec:experimental_design}

\subsection{Challenges when evaluating AI and physics-based weather forecasts}

Table~\ref{tab:evaluation_design} lists several design choices that arise when comparing AI and physics-based weather forecasts. 
AI weather models typically provide forecasts for many variables, at multiple times of the day, on a spatial grid over the entire globe, extending out to at least 10 days in advance. Drawing general conclusions about whether AI models perform ``worse'' or ``better'' than physics-based models is therefore difficult, since different models may perform better for different weather variables, spatial regions, or periods of time.
Cross-study comparisons are further complicated by differences in the analysed variables, spatio-temporal domains, and models, as well as by the continual release of new models and model versions.


When evaluating predictive performance, forecasts are typically compared to gridded weather products that provide an estimate of the weather conditions at each grid point and time point; these weather products are generally large-scale \emph{(re)analysis fields} \citep{Hersbach2020}. 
Past archives of these reanalysis fields are also used to train AI models. 
However, if forecasts are evaluated using the same data on which AI models are trained, then such a comparison unfairly favours AI models, since physics-based models are not designed to reproduce this particular estimate of the atmosphere. 
Moreover, reanalyses are produced retroactively, and may incorporate information unavailable to physics-based models at the time of forecasting. To circumvent this, forecasts can be evaluated using independent observation (or \emph{ground truth}) datasets. Alternatively, \emph{operational} versions of AI models have been proposed that are pretrained on reanalysis fields, before being fine-tuned on weather products that were available to physics-based models at initialisation time. This introduces additional questions regarding which observation datasets and model versions should be compared.


Similarly, AI models are trained to optimise a loss function, such as the mean squared error (MSE) or mean absolute error (MAE). Squared-error and absolute-error losses encourage predictions of conditional means and medians, respectively \citep{Gneiting2011}. 
If forecasts are evaluated using the same loss function used to train AI models, then AI models again have an unfair advantage over physics-based models, which do not optimise a specific loss function. Since the mean and median represent measures of central tendency, training AI models to optimise the MSE or MAE can also yield spatially smooth forecast fields \citep{Bonavita2024}. In contrast, deterministic physics-based models provide a single possible realisation of the future weather. 
Comparing forecasts that target different functionals is not straightforward within standard forecast evaluation theory.
Several papers have proposed circumventing this by converting the forecasts to probabilistic forecasts, and then evaluating the probabilistic forecasts using proper scoring rules \citep{Gneiting2026pcrps,Brenowitz2025,Loveday2025}. However, doing so ignores some information in the forecast.

\begin{table}[t]
    \caption{Examples of design choices when comparing AI and physics-based weather models.}
    \label{tab:evaluation_design}
    \centering
    \small
    \renewcommand{\arraystretch}{1.18}
    \setlength{\tabcolsep}{4pt}
    \begin{tabularx}{\linewidth}{
        @{}
        >{\raggedright\arraybackslash}p{0.23\linewidth}
        >{\raggedright\arraybackslash}X
        @{}
    }
        \toprule
        Design choice & Possible specifications \\
        \midrule
        Forecast systems
        & Physics- or AI-based; deterministic or probabilistic; model versions \\

        Observation data
        & Operational analyses, reanalyses,  or in-situ observations \\

        Spatial resolution
        & Native or re-gridded \\

        Weather variables
        & Surface, upper-air, or derived; instantaneous, extrema, or accumulated \\

        Spatio-temporal domain
        & Global, land, ocean, regional, or several stations; different years, times of day \\

        Evaluation methods
        & Scoring functions, scoring rules, contingency tables, other measures \\

        Definition of extremes
        & Records, quantile-based, absolute threshold-based \\

        \bottomrule
    \end{tabularx}
\end{table}

\subsection{Challenges when evaluating forecasts of extreme events}

Comparing predictions made for extreme events is particularly challenging. Firstly, 
different studies have employed different definitions: \cite{OlivettiMessori2024} and \cite{Loveday2025} study extremes defined as exceedances of high quantiles (e.g., of historical observation data), while \citet{Zhang2026} and \cite{Biegert2026twpcrps} focus on record-breaking events, defined as values above historical maxima or below historical minima. 
The dependence of the definitions of extremes on space and time also varies between studies.

Secondly, theoretical challenges arise when evaluating predictions of extremes. 
An intuitive approach is to evaluate only the forecasts issued when an extreme event occurs, as in \cite{OlivettiMessori2024} and \cite{Zhang2026}. 
However, this approach incentivises more extreme predictions, since it does not penalise forecasts that incorrectly predict that an extreme event will occur. \cite{Lerch2017dilemma} term this the forecaster's dilemma. While one could instead seek scoring functions that are consistent for some extreme-value statistic, such as the extreme value index, \cite{BrehmerStrokorb2019} show that such scoring functions often do not exist. \cite{Taggart2022} proposes a weighted version of the squared error designed to target threshold exceedances, though this is still consistent for the mean functional, and thus arguably also favours AI models trained using a squared-error loss. 
These considerations motivate our systematic investigation of how evaluation metrics affect the ranking of AI and physics-based weather models.

\subsection{Our design and method choices}
\label{sec:experimental_design:subsec:choices}

While some of these challenges are unavoidable, in Section \ref{sec:results} we compare AI and physics-based weather forecasts for extremes in a common experimental design to identify the reasons for the contrasting results in previous studies. We compare the operational variants of three AI models, Pangu-Weather \citep{Bi2023pangu}, GraphCast \citep{Lam2023graphcast}, and Aurora \citep{Bodnar2025aurora}, with the ECMWF's high resolution (HRES) physics-based model. The forecasts are evaluated against the analysis fields used to initialise all four models. 
We evaluate all models on a common $0.25^\circ \times 0.25^\circ$ grid at 708 initialisation times at 00 and 12 UTC in 2022, a year not included in the pretraining or fine-tuning data of any AI model, and restrict the lead times to 1, 3, 5, 7, and 10 days. 
Following~\citet{Zhang2026}, we consider predictions for surface temperature and wind speed, and define extremes as record-breaking heat, cold, and wind events, where the records are defined separately for each weather variable, grid point, and month, using historical ERA5 reanalysis data from 1979--2017 \citep{Hersbach2020}. Results are aggregated across all grid points over land, excluding the Antarctic region; global results are reported as a robustness check in Appendix \ref{app:global}. All methods are therefore compared using the same forecast--observation pairs, and the same record thresholds. The four models are compared using the mean squared error, the potential CRPS \citep{Gneiting2026pcrps}, and three evaluation metrics designed to emphasise extreme events.

\section{Evaluation methods and score decompositions}
\label{sec:eval_and_decomp}

A forecast model can be evaluated using pairs of predictions and corresponding observations, denoted by $(x_1, y_1), \dots, (x_n, y_n)$. A prediction $x \in \R$ for $y \in \R$ can be evaluated using a scoring function $s : \R \times \R \to [0, \infty)$, where the score $s(x, y)$ quantifies the error in the prediction; we adopt the convention that a lower score indicates a better prediction. 
Forecast models can then be evaluated using the average score $\sum_{i=1}^n s(x_i, y_i)/n$ over the evaluation data. Average scores are calculated over all grid points and forecast initialisation times. 
While weather forecasts are considered for multiple models, weather variables, and forecast horizons (or lead times), this is suppressed in the notation for simplicity. The scores at each grid point are weighted by the spatial area of each grid point; this is also omitted from the following notation, and details can be found in Appendix~\ref{app:spatial_aggregation}.

\subsection{Evaluation methods}
\label{sec:eval_and_decomp:subsec:eval}

\subsubsection{Mean Squared Error}

The most commonly used scoring function to evaluate the prediction $x$ is the squared error, $s(x, y) = (x - y)^2$, and forecast models are typically compared with respect to their mean squared error $\operatorname{MSE} = \sum_{i=1}^n (x_i - y_i)^2 / n$. When predicting the unknown target variable $Y \in \R$, the expected squared error, $\E[s(x, Y)]$, is (uniquely) minimised when $x = \E[Y]$. In other words, the squared error is (strictly) \emph{consistent} for the mean functional \citep[see][]{Gneiting2011}. The MSE is also commonly used as a loss function to train AI weather models, and since this encourages predictions for $\E[Y]$, this arguably yields predictions that underpredict the occurrence of extreme events \citep{Rasp2024weatherbench2}.

\subsubsection{Outcome-weighted MSE}

For thresholds $t_1, \dots, t_n \in \R$, let $r_i=\mathbbm{1}\{y_i>t_i\}$ when interest is on extremely large outcomes (e.g. for heat and wind records), and let $r_i=\mathbbm{1}\{y_i<t_i\}$ when the focus is on extremely low outcomes (e.g. for cold records). We assume that $\sum_{i=1}^n r_i>0$. When interest is on extremes, a seemingly attractive strategy is to calculate the forecast error conditional on an extreme event having occurred. For the MSE, we refer to the resulting score as the \emph{outcome-weighted MSE}, defined as
\begin{equation}
    \operatorname{owMSE}
    =
    \frac{1}{\sum_{i=1}^n r_i} \sum_{i=1}^n r_i (x_i - y_i)^2.
     \label{eq:owmse}
\end{equation}
However, while intuitive, the owMSE is not consistent for the mean functional; for exceedances of a fixed threshold $t$, the expected score is instead minimised by $x = \E[Y \mid Y > t]$. This raises the forecaster's dilemma, since forecasts are encouraged to predict extreme events, regardless of how likely they are to occur. In particular, the owMSE does not penalise \emph{false alarms}, where the forecast incorrectly predicts that an extreme event will occur.

\subsubsection{Threshold-weighted MSE}

To overcome this limitation, \citet{Taggart2022} proposed to separate scoring functions into component scores that emphasise different regions of the outcome space. For example, given a threshold $t \in \R$, the squared error can be decomposed as $(x-y)^2 = s_{1,t}(x,y)+s_{2,t}(x,y)$,
where 
\begin{align*}
    s_{1,t}(x,y) &= (t-y)_+^2 - (t-x)_+^2 - 2(x-y)(t-x)_+, \\
    s_{2,t}(x,y) &= (y-t)_+^2 - (x-t)_+^2 - 2(y-x)(x-t)_+,
\end{align*}
with $(z)_+=\max\{z,0\}$. The scoring function $s_{1,t}$ focuses on the region $(-\infty,t)$, while $s_{2,t}$ focuses on $[t,\infty)$. Both $s_{1,t}$ and $s_{2,t}$ are consistent for the mean functional. However, they are not strictly consistent, since they are insensitive to forecast performance outside their respective regions of interest. In particular, the expected score $\E[s_{2,t}(x, Y)]$ is minimised by $x = \E[Y]$ when $\E[Y] > t$, and any $x \le t$ when $\E[Y] \le t$, in which case $x = \E[Y] \le t$ is only one possible optimum forecast.

When evaluating heat and wind extremes, we use the score $s_{2,t}$ and compare forecast models using
\begin{equation}
    \operatorname{twMSE}
    =
    \frac{1}{n}\sum_{i=1}^n
    s_{2,t_i}(x_i, y_i),
    \label{eq:twmse}
\end{equation}
which we refer to as the \emph{threshold-weighted MSE}. For cold extremes, $s_{1,t_i}$ is used in place of $s_{2,t_i}$. The twMSE is sensitive to hits (when the forecast correctly predicts that an extreme event will occur), misses (when the forecast fails to predict that an extreme event will occur), and false alarms; correctly predicting that no extreme event will occur contributes zero to the twMSE. However, while the owMSE ignores false alarms, the twMSE penalises them three times as much as misses when forecasts and observations are symmetric about the threshold, thus potentially encouraging less extreme predictions (see Appendix~\ref{app:twmse_details}).

\subsubsection{PCRPS and threshold-weighted PCRPS}
Evaluating forecasts using the MSE implicitly favours AI models trained using a squared-error loss. To circumvent this, \cite{Gneiting2026pcrps} propose to recalibrate deterministic weather forecasts to yield probabilistic forecasts, and to evaluate the resulting forecasts using proper scoring rules, such as the continuous ranked probability score (CRPS). When interest is on extreme events, \cite{Biegert2026twpcrps} extend this approach to employ weighted scoring rules that focus the evaluation on extremes. 
More specifically, for each model, variable, grid point, and lead time, isotonic distributional regression \citep[IDR, or EasyUQ;][]{Walz2024easyuq,Henzi2021idr} is applied to the pairs $(x_1, y_1), \dots, (x_n, y_n)$, converting each point forecast $x_i$ to a predictive distribution function $\widehat F_i$ (under a stochastic ordering constraint). The \emph{potential CRPS} is the average CRPS of these predictive distributions. 

When interest is on threshold exceedances, we transform the pairs $(x_i, y_i)$ to exceedances of the respective thresholds, $(\widetilde x_i,\widetilde y_i)=(x_i-t_i,y_i-t_i)$ for heat and wind records, and $(\widetilde x_i,\widetilde y_i)=(t_i-x_i,t_i-y_i)$ for cold records. For each record type, IDR is then fit to these transformed pairs, yielding predictive distribution functions $\widehat G_i$. We then compute the \emph{threshold-weighted potential CRPS} as
\begin{align}
    \operatorname{twPCRPS}
    =
    \frac{1}{n}\sum_{i=1}^n
    \int_0^\infty
    \left(\widehat G_i(z)-\mathbbm{1}\{\widetilde{y}_i\leq z\}\right)^2
    \,\mathrm{d}z. 
    \label{eq:twpcrps}
\end{align}
The IDR recalibrated distributions $\widehat{G}_i$ are constructed to minimise this average score subject to the stochastic ordering constraint on the transformed forecasts \citep[Theorem 2]{Henzi2021idr}. By focusing on the interval $(0, \infty)$, the twPCRPS assesses the forecasts via their ability to predict threshold exceedances. The IDR recalibration removes biases in the evaluation that arise due to differences in the loss function and datasets used to train AI models; however, this approach also corrects for miscalibration due to the model itself, and it therefore only measures the \emph{potential} predictive ability of the forecasts. The transformation from $(x_i, y_i)$ to $(\widetilde x_i,\widetilde y_i)$ ensures that the same threshold, and therefore the same scoring rule, is applied to all forecast cases, which accounts for the monthly varying thresholds used to define extreme events. If the thresholds $t_i$ are the same for all forecast cases, then this transformation has no effect on the resulting twPCRPS values, and it can therefore be ignored. Further details are given in Appendix~\ref{app:twpcrps_details}.

\subsection{Score decompositions}
\label{sec:score_decompositions}

The average score provides a single numerical value with which competing prediction systems can be compared. Alone, this score does not explain why one prediction system receives a better score than another. However, the average score can be decomposed into terms that quantify the forecast uncertainty, forecast information content, and forecast miscalibration, allowing practitioners to study the contributions of these components to the total score. In particular, for any scoring function $s$, 
\begin{equation}\label{eq:decomp_samp}
        \frac{1}{n} \sum_{i=1}^n s(x_i, y_i) 
    = \underbrace{\frac{1}{n} \sum_{i=1}^n s(\bar{y}, y_i)}_{\mathrm{\widehat{UNC}}} 
    - \underbrace{\frac{1}{n} \sum_{i=1}^n \left[ s(\bar{y}, y_i) - s(x_i^{(r)}, y_i) \right]}_{\mathrm{\widehat{DSC}}} 
    + \underbrace{ \frac{1}{n} \sum_{i=1}^n \left[ s(x_i, y_i) - s(x_i^{(r)}, y_i) \right]}_{\mathrm{\widehat{MCB}}},
\end{equation}
where $\bar{y} = \sum_{i=1}^n y_i /n$ and $x_i^{(r)}$ is a recalibrated version of the prediction $x_i$. The uncertainty component ($\mathrm{\widehat{UNC}}$) 
quantifies the inherent uncertainty in the target variable; for the squared error, this is simply the sample variance of $y_1, \dots, y_n$. The discrimination component ($\mathrm{\widehat{DSC}}$) quantifies how much the observations $y_1, \dots, y_n$ depend on the predictions $x_1, \dots, x_n$, providing a measure of the prediction's information content, or discriminative ability. The miscalibration term ($\mathrm{\widehat{MCB}}$) quantifies how much the predictions differ from their recalibrated versions $x_1^{(r)}, \dots, x_n^{(r)}$.

\citet{Dimitriadis2021} and \citet{GneitingResin2023} suggest recalibrating the predictions by fitting an isotonic (mean) regression model for $y_1, \dots, y_n$ using $x_1, \dots, x_n$ as the covariates. When the isotonic regression model is learned using the standard pool-adjacent-violators (PAV) algorithm, fit in-sample on the evaluation data, the recalibrated predictions $x_1^{(r)}, \dots, x_n^{(r)}$ are guaranteed to perform at least as well as the original predictions $x_1, \dots, x_n$ and the mean prediction $\bar{y}$ when evaluated using any scoring function that is consistent for the mean functional. In this case, the miscalibration and discrimination terms are necessarily non-negative \citep[Theorem 3.3]{GneitingResin2023}.
For strictly consistent scoring functions, the discrimination term is equal to zero if and only if $x_i^{(r)} = \bar{y}$, for all $i = 1, \dots, n$, suggesting that $x_i$ provides no information about $y_i$, and the miscalibration term is equal to zero if and only if $x_i^{(r)} = x_i$, suggesting the original forecasts are calibrated. 

The decomposition can be applied using the owMSE, twMSE, and twPCRPS. Since these scoring functions all focus on extreme outcomes, the resulting miscalibration components assess forecast calibration in the relevant tail. This allows us to assess whether forecasts are calibrated when predicting record breaking events. However, since the owMSE restricts the evaluation to predictions issued when an extreme event occurs, the resulting decomposition terms are not meaningful. An analogous decomposition can also be applied to the twPCRPS using methods in \cite{ArnoldEtAl2024}, but since they also propose to recalibrate predictions using IDR, the resulting $\mathrm{\widehat{MCB}}$ term will be zero for all predictions by construction. This emphasises that the PCRPS and twPCRPS focus only on the discriminative ability of the predictions, and it should therefore be expected that comparisons in terms of the twPCRPS are similar to comparisons of the discrimination term of the twMSE. 

We therefore restrict attention to the decomposition terms calculated using the twMSE. For the twMSE, which is consistent but not strictly consistent for the mean functional, the miscalibration component assesses a notion of calibration that is specific to mean forecasts, whilst restricting attention to the tails; these theoretical properties are provided in Appendix~\ref{app:score_decomp}. When predictions take the form of samples from the outcome distribution, as is the case for physics-based weather models, this definition of calibration may not necessarily be a property that we want to be satisfied. Similarly, Appendix \ref{app:twmse_details} shows the asymmetric treatment of false alarms and misses by the twMSE. Hence, while a miscalibration term equal to zero suggests that the forecasts are calibrated when predicting extremes, different types of miscalibrations will be penalised differently when the forecasts are not calibrated. This should be taken into account when interpreting the decomposition terms in the following section. For comparison, decomposition terms corresponding to the owMSE are shown in Appendix~\ref{app:additional_results} (Figures \ref{fig:owmse_decomposition_0p25_land} and \ref{fig:relative_miscalibration_0p25_land}). For the monthly record thresholds, the twMSE decomposition is applied to the same transformed pairs $(\widetilde x_i,\widetilde y_i)$ used for the twPCRPS, as detailed in Appendix~\ref{app:twmse_details}.

\section{Results}
\label{sec:results}

\subsection{Model rankings across evaluation metrics}

Figure~\ref{fig:main_scores_0p25_land} compares model rankings across evaluation metrics. Panel~(a) reports the unweighted MSE and PCRPS over all forecast cases. At day 1, the physics-based model (HRES) attains a marginally lower score than all AI models for both weather variables and metrics. At longer lead times, Aurora issues the most accurate temperature predictions, followed by GraphCast. For wind speed, Aurora has the lowest PCRPS after day 1, while GraphCast receives the lowest MSE when forecasting more than one week ahead. Across both variables, Pangu-Weather generally performs worst.

\begin{figure}
  \centering
  \includegraphics[width=\linewidth]{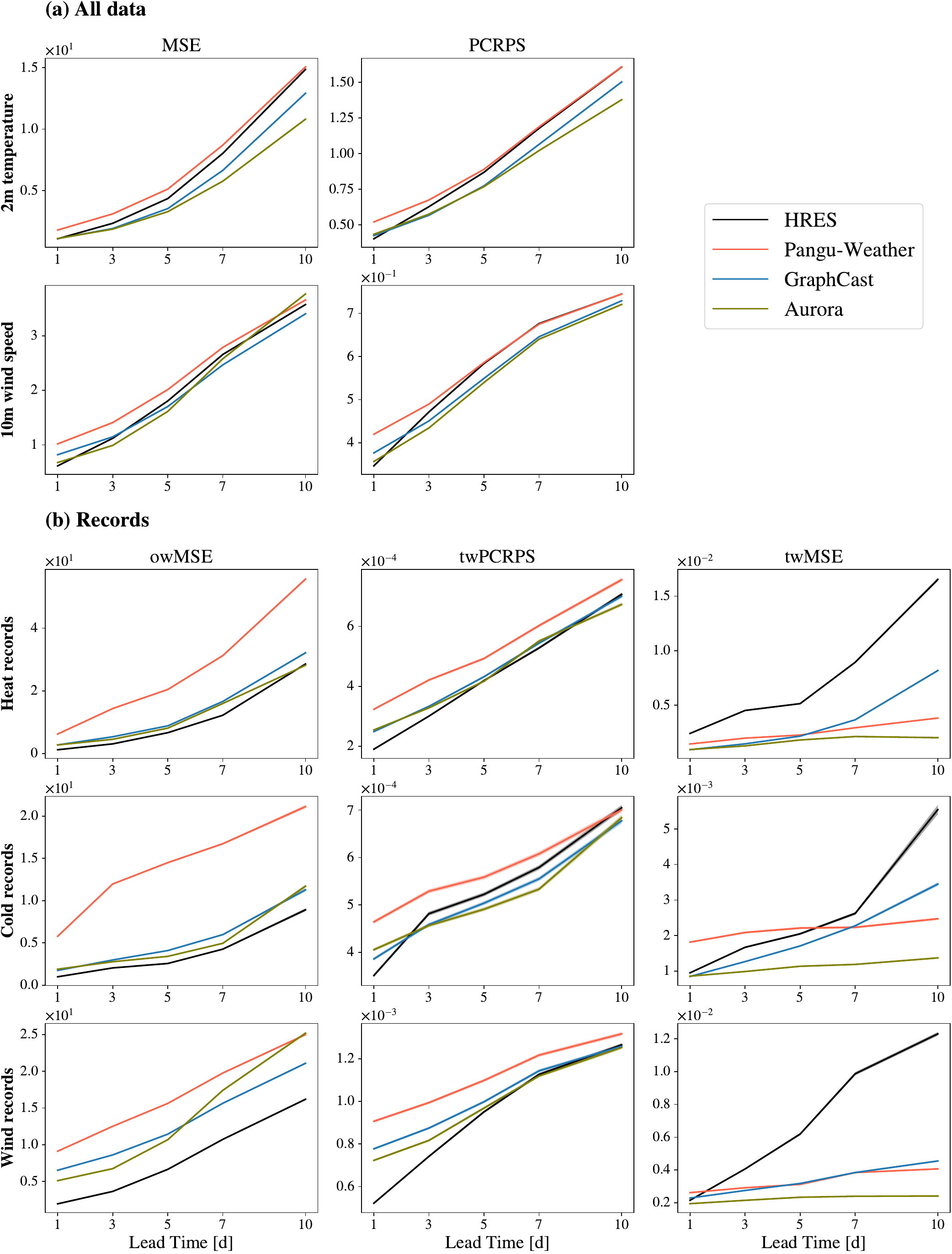}
  \caption{
    {Average scores for the four models.
    Panel~(a) shows MSE and PCRPS over all forecast cases. Panel~(b) shows owMSE, twPCRPS, and twMSE for heat, cold, and wind records. Shaded regions denote pointwise 99\% confidence intervals, calculated as in Appendix~\ref{app:confidence}. Lower values indicate better performance. MSE, owMSE, and twMSE are reported in $\mathrm{K}^2$ for temperature and $\mathrm{m}^2\,\mathrm{s}^{-2}$ for wind speed; PCRPS and twPCRPS are reported in $\mathrm{K}$ and $\mathrm{m}\,\mathrm{s}^{-1}$, respectively.}
  }
  \label{fig:main_scores_0p25_land}
\end{figure}

The differences between evaluation metrics are considerably more pronounced when evaluating predictions for extreme events in panel~(b). With the owMSE, we recover the main qualitative finding of \cite{Zhang2026}: HRES generally receives the lowest score for heat, cold, and wind records across all lead times, while Pangu-Weather generally receives the highest score. Under the twPCRPS, HRES no longer consistently performs better than the AI models, aligning with results in \cite{Biegert2026twpcrps}. Aurora receives the lowest twPCRPS for heat records at day 10. At medium to long lead times, Aurora and GraphCast outperform HRES for cold records, while differences between models become smaller over time for wind records. Pangu-Weather also generally receives the highest twPCRPS. In stark contrast to the owMSE, the physics-based model performs considerably worse than all AI models when evaluated using the twMSE at longer lead times. Aurora generally receives the lowest score for all three record types, and despite its poor owMSE, Pangu-Weather also frequently has a lower twMSE than HRES and GraphCast at longer lead times. Consequently, the conclusion that HRES performs better for record-breaking events appears to be specific to the owMSE. 

The reason for this result is that, by conditioning on observed record-breaking events, the owMSE ignores false alarms, so forecasts are penalised for missing extreme events, but not for predicting extremes that do not occur. Since AI weather models are typically trained to optimise the MSE or MAE, they are encouraged to issue predictions of conditional means or medians, and they thus tend to forecast extreme weather events with lower frequency than physics-based models. 
AI models therefore perform badly when assessed using the owMSE, which only penalises errors related to hits and misses. In contrast, the twMSE additionally accounts for false alarms, and it penalises these more heavily than misses. Hence, when evaluated using the twMSE, the physics-based model is heavily penalised for predicting that extreme events will occur when they do not. The contribution of misses and false alarms to the scores is examined in detail in the following subsection.

\subsection{Record-occurrence behavior}

Table~\ref{tab:record_contingency_land} reports contingency table statistics for each model when predicting the occurrence of an extreme event: the \emph{miss rate} is $\mathrm{MR}=M/(H+M)$, the \emph{false alarm ratio} is $\mathrm{FAR}=F/(H+F)$, and the frequency bias is $\mathrm{FB} = (H+F)/(H+M)$, where $H$, $M$, and $F$ denote counts of hits, misses, and false alarms, respectively, pooled over all forecast cases. $\mathrm{MR}$ and $\mathrm{FAR}$ equal one minus recall and precision, respectively, while $\mathrm{FB}<1$ ($>1$) indicates under- (over-) prediction. 

\begin{table}[t]
  \caption{
    {Miss rate (MR), false alarm ratio (FAR), and frequency bias (FB) for the four models and three types of record-breaking events. The percentage contributions of hits, misses, and false alarms to twMSE are also shown, with the largest contribution in each row in bold. Results are pooled over all initialisation times, grid points, and lead times. The twMSE contributions are latitude weighted.}
  }
  \label{tab:record_contingency_land}
  \centering
    \setlength{\tabcolsep}{5pt}
    \renewcommand{\arraystretch}{1.08}
    \begin{tabular}{@{}llrrr@{\hspace{14pt}}rrr@{}}
    \toprule
    &
    & \multicolumn{3}{c}{Contingency statistics}
    & \multicolumn{3}{c}{twMSE share [\%]} \\
    \cmidrule(lr){3-5}
    \cmidrule(lr){6-8}
    Record type
    & Model
    & MR [\%]
    & FAR [\%]
    & FB
    & Hit
    & Miss
    & False alarm \\
    \midrule
    
    \multirow{4}{*}{Heat}
    & HRES          & 46.89 & 65.09 & 1.521 &  7.77 &  6.50 & \textbf{85.72} \\
    & Pangu-Weather & 81.62 & 49.77 & 0.366 &  9.07 & \textbf{59.86} & 31.07 \\
    & GraphCast     & 69.93 & 49.76 & 0.599 & 13.25 & 27.98 & \textbf{58.77} \\
    & Aurora        & 73.16 & 29.41 & 0.380 & 22.32 & \textbf{58.39} & 19.30 \\

    \addlinespace
    \multirow{4}{*}{Cold}
    & HRES          & 48.68 & 45.68 & 0.945 & 16.28 & 18.61 & \textbf{65.10} \\
    & Pangu-Weather & 94.46 & 42.64 & 0.097 &  8.91 & \textbf{82.45} &  8.63 \\
    & GraphCast     & 69.70 & 37.05 & 0.481 & 24.93 & \textbf{39.89} & 35.18 \\
    & Aurora        & 64.73 & 22.25 & 0.454 & 33.24 & \textbf{56.78} &  9.99 \\

    \addlinespace
    \multirow{4}{*}{Wind}
    & HRES          & 62.30 & 62.06 & 0.994 &  7.17 & 15.39 & \textbf{77.44} \\
    & Pangu-Weather & 95.79 & 76.45 & 0.179 &  3.53 & \textbf{68.27} & 28.20 \\
    & GraphCast     & 91.26 & 68.10 & 0.274 &  7.86 & \textbf{59.00} & 33.14 \\
    & Aurora        & 90.23 & 35.95 & 0.153 & 12.16 & \textbf{82.18} &  5.66 \\

    \bottomrule
  \end{tabular}
\end{table}
The physics-based model has the lowest miss rate for all three record types. Its frequency bias is close to one for cold and wind records, but its value of $1.521$ for heat records indicates substantial overprediction. All three AI models instead underpredict the occurrence of record-breaking events. Among them, Aurora has the lowest $\mathrm{FAR}$ for every record type and the lowest $\mathrm{MR}$ for cold and wind records, while GraphCast has a slightly lower miss rate for heat records. 

Table~\ref{tab:record_contingency_land} additionally shows the proportion of twMSE attributable to hits, misses, and false alarms. These are calculated as described in Appendix~\ref{app:twmse_details}. False alarms account for $65.1\%$--$85.7\%$ of the HRES twMSE across the three record types, whereas the twMSE of the AI models is generally dominated by misses. This is particularly true for Aurora and Pangu-Weather, while the differences between miss and false alarm contributions are smaller for GraphCast. 
By neglecting false alarms, the owMSE ignores the dominant source of error in the twMSE of HRES, helping to explain the contrasting model rankings between the different evaluation metrics.

\subsection{Miscalibration and discrimination}
By recalibrating the predictions, the twPCRPS is less sensitive to the proportions of misses and false alarms, instead providing a measure of information content in the predictions. To verify this, Figure~\ref{fig:twmse_decomposition_0p25_land} shows the estimated miscalibration and discrimination components of the twMSE decomposition.  
As expected, forecast discrimination decreases with lead time for all models and record types, whereas miscalibration generally increases. The discrimination component is generally similar for all models, suggesting all models contain roughly the same information when forecasting extreme weather events; Pangu-Weather consistently issues the least informative predictions.

Instead, model differences are largely attributable to miscalibration. The miscalibration component of HRES increases sharply with lead time, becoming much larger than that of all AI models. This suggests that HRES does not issue calibrated forecasts for extreme weather events, where calibration is defined as in Appendix \ref{app:score_decomp_twmse}. This is related to how the twMSE penalises false alarms, and explains the poor relative performance of the physics-based model when assessed using the twMSE. Among the AI models, Aurora tends to issue the most calibrated forecasts for all record types.

%
This also explains why HRES performs more favourably when assessed using the twPCRPS than the twMSE. As expected, the twPCRPS rankings in Figure~\ref{fig:main_scores_0p25_land}(b) broadly align with the discrimination component in Figure~\ref{fig:twmse_decomposition_0p25_land}, particularly for cold and wind records. 
This indicates that differences in model rankings between the twPCRPS and twMSE are largely attributable to forecast miscalibration. 

\begin{figure}
  \centering
  \includegraphics[width=\linewidth]{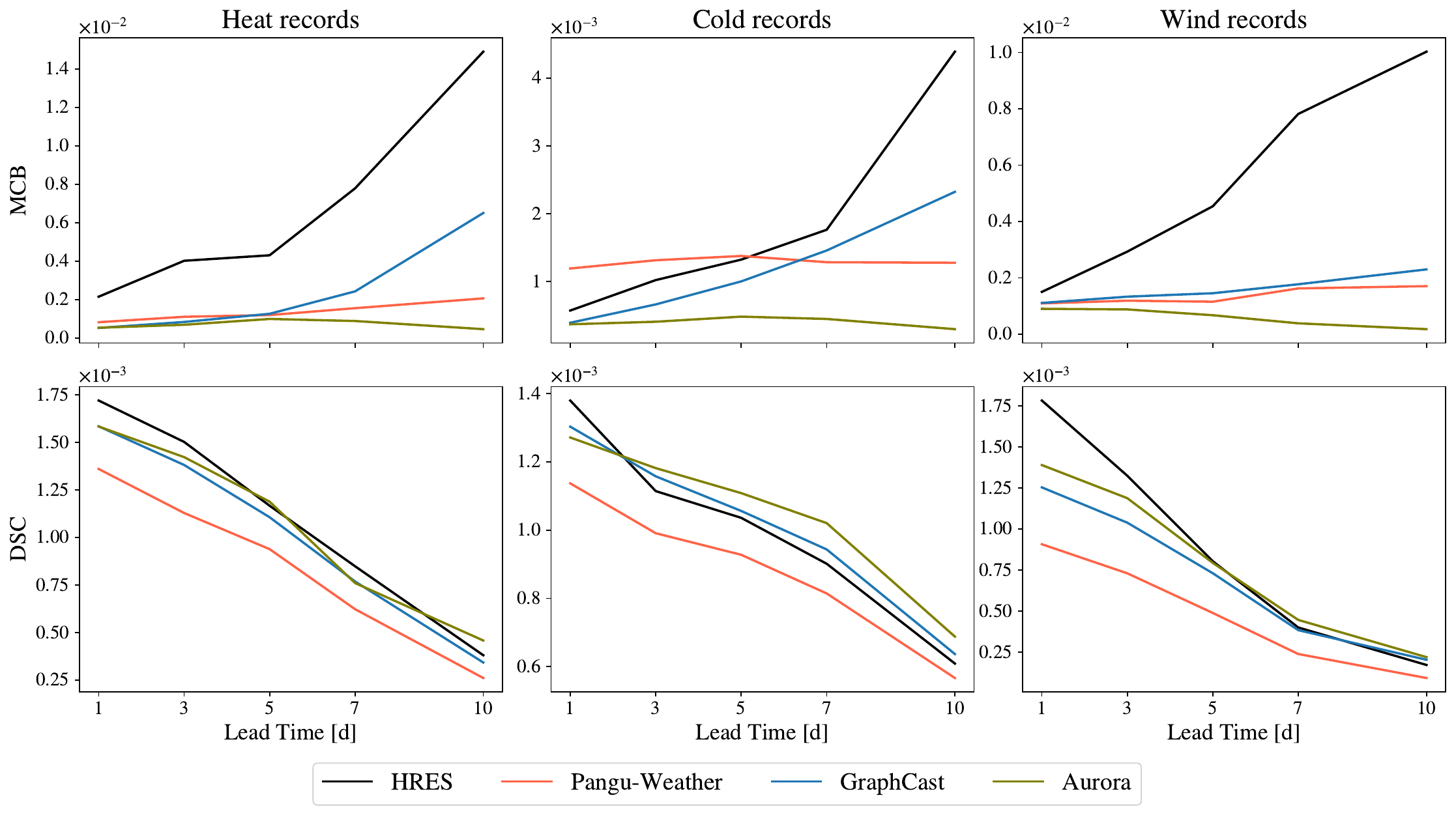}
  \caption{
    {\small 
    Miscalibration (MCB) and discrimination (DSC) components of the decomposition of the average twMSE when evaluating predictions for heat, cold, and wind records. Lower MCB indicates better calibrated predictions, higher DSC indicates more informative predictions. Note the different scales on the $y$-axes.}
  }
  \label{fig:twmse_decomposition_0p25_land}
  \vspace*{-0.5cm}
\end{figure}

\section{Conclusion}\label{sec:conc}

We find that the ranking of physics-based and AI weather models when forecasting extreme events depends strongly on the chosen evaluation metric, even under a common experimental design. 
HRES, the physics-based model investigated here, forecasts record-breaking extreme weather events with a much higher frequency than AI models, which are implicitly designed to issue predictions for the conditional mean or median. If forecasts are evaluated conditional on an extreme event having occurred, as in \cite{Zhang2026}, then physics-based models typically have an advantage. In contrast, 
AI models are generally judged to issue more accurate forecasts than physics-based models when false alarms are penalised more heavily than misses. By decomposing average scores, all models are found to have similar information content, and differences between rankings are largely attributable to forecast miscalibration, with different scores penalising different miscalibrations more or less heavily. This explains contrasting conclusions about the performance of AI weather models when forecasting extreme events, and highlights the importance of understanding what exactly the evaluation metric assesses when comparing AI models.

\section{Limitations}\label{sec:limitations}

In Section \ref{sec:experimental_design}, we list several challenges that impede a fair comparison of AI and physics-based weather models. While we attempt to circumvent these challenges, some are unavoidable. For example, the AI models that we compare are trained on slightly different data sets. In particular, the Aurora model saw observation data between 2018 and 2020 during pretraining, which is outside the record archive. It therefore sees record-breaking events already during pretraining. 
However, since both GraphCast and Aurora are fine-tuned using data beyond the record archive, it is unlikely that this has a significant impact on the conclusions. Furthermore, our empirical comparison covers three AI models and one physics-based model for record-breaking heat, cold, and wind events in 2022. Future work should therefore investigate whether the observed calibration patterns and model rankings generalise to other forecasting systems, evaluation periods, and definitions of extreme events.

Similarly, we evaluate all models against the same observation fields, providing a common evaluation framework. However, these observation fields are intrinsically connected to the construction of the HRES model, and this choice of evaluation data therefore arguably slightly favours the physics-based forecasts. These observation fields are also subject to uncertainties and errors that we do not account for here; these errors can affect scores and the classification of record breaking events, with potentially different effects across models. The robustness of our findings could be assessed by performing the same analysis using an independent set of observation data, such as recordings at weather stations. 

Uncertainty in the results is represented using confidence intervals around average scores in Figure \ref{fig:main_scores_0p25_land}. The method used to calculate these confidence intervals is detailed in Appendix \ref{app:confidence}. This approach assumes that the individual scores are independent over time and space, which is not always realistic. Note, however, that this is much weaker than assuming that the observed weather values are themselves independent over time and space, since the forecast should account for a large amount of this dependence. Performing paired significance tests would also partially address this. Accounting for this dependence would most likely increase the spread of the confidence intervals, and conclusions about significance should therefore be interpreted with caution. 
We additionally do not consider the uncertainty in the score decomposition terms presented in Figure \ref{fig:twmse_decomposition_0p25_land}. Methods for inference on score decomposition terms are sparse \citep{DimitriadisPuke2026}, particularly when the decomposition terms are estimated using isotonic regression.

Finally, following \cite{Lerch2017dilemma}, we argue that the owMSE, which restricts attention to predictions issued when extreme events have occurred, favours models that overpredict the occurrence of extremes. In contrast, the twMSE of \cite{Taggart2022} emphasises threshold exceedances whilst still being a consistent scoring function for the mean. While this arguably gives the twMSE a more solid theoretical grounding than the owMSE, it also arguably inadvertently favours AI models trained using a squared-error loss over physics-based models. The twMSE is also found to penalise false alarms more heavily than misses. Hence, in applications where missing a hazardous event is more costly than issuing a false alarm, the twMSE may not reflect the user-specific preferences. The twPCRPS of \cite{Biegert2026twpcrps} instead recalibrates the forecasts prior to evaluation, but in doing so, it only assesses the forecast information content. It is an open question whether other evaluation metrics exist that provide a suitable means with which to evaluate deterministic forecasts for extreme events.



\section*{Data and code availability}
The WeatherBench 2 data used in this study are publicly available; details on data access are provided by \citet{Rasp2024weatherbench2}. Code to reproduce the results presented in this paper is available at \url{https://github.com/tobiasbiegert/weather-extremes-evaluation}.

\begin{ack}
Tobias Biegert gratefully acknowledges support by the German Weather Service (Deutscher Wetterdienst) through the SPARC-ML project within the extramural research programme, funding reference number 4823EMF01. Sam Allen, Zhongwei Zhang and Sebastian Lerch gratefully acknowledge support by the Vector Stiftung through the Young Investigator Group ``Artificial Intelligence for Probabilistic Weather Forecasting''. 
\end{ack}

\bibliographystyle{plainnat}
\bibliography{references}

@article{AllenEtAl2025,
  title={Tail calibration of probabilistic forecasts},
  author={Allen, Sam and Koh, Jonathan and Segers, Johan and Ziegel, Johanna},
  journal={Journal of the American Statistical Association},
  volume={120},
  pages={2796--2808},
  year={2025},
  publisher={Taylor \& Francis},
  doi={10.1080/01621459.2025.2506194}
}

@article{ArnoldEtAl2024,
  title={Decompositions of the mean continuous ranked probability score},
  author={Arnold, Sebastian and Walz, Eva-Maria and Ziegel, Johanna and Gneiting, Tilmann},
  journal={Electronic Journal of Statistics},
  volume={18},
  pages={4992--5044},
  year={2024},
  publisher={The Institute of Mathematical Statistics and the Bernoulli Society},
  doi={10.1214/24-EJS2316}
}

@article{Zhang2026,
  title={Physics-based models outperform {AI} weather forecasts of record-breaking extremes},
  author={Zhang, Zhongwei and Fischer, Erich and Zscheischler, Jakob and Engelke, Sebastian},
  journal={Science Advances},
  volume={12},
  pages={eaec1433},
  year={2026},
  publisher={American Association for the Advancement of Science},
  doi={10.1126/sciadv.aec1433}
}

@article{OlivettiMessori2024,
  title={Do data-driven models beat numerical models in forecasting weather extremes? {A} comparison of {IFS HRES}, {Pangu-Weather}, and {GraphCast}},
  author={Olivetti, Leonardo and Messori, Gabriele},
  journal={Geoscientific Model Development},
  volume={17},
  pages={7915--7962},
  year={2024},
  publisher={Copernicus Publications G{\"o}ttingen, Germany},
  doi={10.5194/gmd-17-7915-2024}
}

@article{Lerch2017dilemma,
  title={Forecaster's dilemma: Extreme events and forecast evaluation},
  author={Lerch, Sebastian and Thorarinsdottir, Thordis L and Ravazzolo, Francesco and Gneiting, Tilmann},
  journal={Statistical Science},
  number = {1},
  volume = {32},
  doi = {10.1214/16-STS588},
  pages={106--127},
  year={2017},
  publisher={Institute of Mathematical Statistics}
}

@article{Gneiting2026pcrps,
  title={Probabilistic measures afford fair comparisons of {AIWP} and {NWP} model output},
  author={Gneiting, Tilmann and Biegert, Tobias and Kraus, Kristof and Walz, Eva-Maria and Jordan, Alexander I and Lerch, Sebastian},
  journal={Artificial Intelligence for the Earth Systems},
  volume={5},
  pages={e250054},
  year={2026},
  publisher={American Meteorological Society},
  doi={10.1175/AIES-D-25-0054.1}
}

@article{Gneiting2011,
  title={Making and evaluating point forecasts},
  author={Gneiting, Tilmann},
  journal={Journal of the American Statistical Association},
  volume={106},
  pages={746--762},
  year={2011},
  publisher={Taylor \& Francis},
  doi={10.1198/jasa.2011.r10138}
}

@article{GneitingResin2023,
  title={Regression diagnostics meets forecast evaluation: Conditional calibration, reliability diagrams, and coefficient of determination},
  author={Gneiting, Tilmann and Resin, Johannes},
  journal={Electronic Journal of Statistics},
  volume={17},
  pages={3226--3286},
  year={2023},
  publisher={The Institute of Mathematical Statistics and the Bernoulli Society},
  doi={10.1214/23-EJS2180}
}

@article{Biegert2026twpcrps,
  title={Towards Fair Comparisons of {AI}- and Physics-Based Weather Models for Extreme Events via the Weighted Potential {CRPS}},
  author={Biegert, Tobias and Allen, Sam and Alber, Annika and Lerch, Sebastian},
  journal={arXiv preprint arXiv:2606.21170},
  year={2026}
}

@article{Taggart2022,
  title={Evaluation of point forecasts for extreme events using consistent scoring functions},
  author={Taggart, Robert},
  journal={Quarterly Journal of the Royal Meteorological Society},
  volume={148},
  pages={306--320},
  year={2022},
  publisher={Wiley Online Library},
  doi={10.1002/qj.4206}
}

@article{Dimitriadis2021,
  title={Stable reliability diagrams for probabilistic classifiers},
  author={Dimitriadis, Timo and Gneiting, Tilmann and Jordan, Alexander I},
  journal={Proceedings of the National Academy of Sciences},
  volume={118},
  pages={e2016191118},
  year={2021},
  publisher={National Academy of Sciences},
  doi={10.1073/pnas.2016191118}
}

@article{Ehm2016,
  title={Of quantiles and expectiles: consistent scoring functions, {Choquet} representations and forecast rankings},
  author={Ehm, Werner and Gneiting, Tilmann and Jordan, Alexander and Kr{\"u}ger, Fabian},
  journal={Journal of the Royal Statistical Society Series B: Statistical Methodology},
  volume={78},
  pages={505--562},
  year={2016},
  publisher={Oxford University Press},
  doi={10.1111/rssb.12154}
}

@article{Henzi2021idr,
  title={Isotonic distributional regression},
  author={Henzi, Alexander and Ziegel, Johanna F and Gneiting, Tilmann},
  journal={Journal of the Royal Statistical Society Series B: Statistical Methodology},
  volume={83},
  pages={963--993},
  year={2021},
  publisher={Oxford University Press},
  doi={10.1111/rssb.12450}
}

@article{Walz2024easyuq,
  title = {Easy {{Uncertainty Quantification}} ({{EasyUQ}}): Generating Predictive Distributions from Single-Valued Model Output},
  shorttitle = {Easy {{Uncertainty Quantification}} ({{EasyUQ}})},
  author = {Walz, Eva-Maria and Henzi, Alexander and Ziegel, Johanna and Gneiting, Tilmann},
  year = 2024,
  journal = {SIAM Review},
  volume = {66},
  pages = {91--122},
  doi={10.1137/22M1541915}
}

@article{BrehmerStrokorb2019,
  title={Why scoring functions cannot assess tail properties},
  author={Brehmer, Jonas R and Strokorb, Kirstin},
  year={2019},
  journal = {Electronic Journal of Statistics},
  volume = {13},
  pages = {4015--4034},
  publisher = {{Institute of Mathematical Statistics and Bernoulli Society}},
  doi={10.1214/19-EJS1622}
}

@article{BenBouallegue2024,
  title={The rise of data-driven weather forecasting: A first statistical assessment of machine learning--based weather forecasts in an operational-like context},
  author={Ben Bouall{\`e}gue, Zied and Clare, Mariana CA and Magnusson, Linus and Gascon, Estibaliz and Maier-Gerber, Michael and Janou{\v{s}}ek, Martin and Rodwell, Mark and Pinault, Florian and Dramsch, Jesper S and Lang, Simon TK and others},
  journal={Bulletin of the American Meteorological Society},
  volume={105},
  pages={E864--E883},
  year={2024},
  publisher={American Meteorological Society},
  doi={10.1175/BAMS-D-23-0162.1}
}

@article{DeMaria2025,
  title={An operations-based evaluation of tropical cyclone track and intensity forecasts from artificial intelligence weather prediction models},
  author={DeMaria, Mark and Franklin, James L and Chirokova, Galina and Radford, Jacob and DeMaria, Robert and Musgrave, Kate D and Ebert-Uphoff, Imme},
  journal={Artificial Intelligence for the Earth Systems},
  volume={4},
  pages={240085},
  year={2025},
  publisher={American Meteorological Society},
  doi={10.1175/AIES-D-24-0085.1}
}

@article{Dacre2026,
  title={Northern hemisphere midlatitude cyclone intensity biases in machine learning weather prediction models},
  author={Dacre, HF and Charlton-Perez, AJ and Driscoll, S and Gray, SL and Harvey, B and Harvey, NJ and Hodges, KI and Hunt, KMR and Volont{\'e}, A},
  journal={Bulletin of the American Meteorological Society},
  volume={107},
  pages={E208--E221},
  year={2026},
  publisher={American Meteorological Society},
  doi={10.1175/BAMS-D-25-0129.1}
}

@article{Davis2026,
  title={Physics-based versus {AI} Weather prediction models: A comparative performance assessment of Atmospheric River prediction},
  author={Davis, Isaac W and Subramanian, Aneesh and Higgins, Timothy B and Sengupta, Agniv and Delle Monache, Luca},
  journal={Geophysical Research Letters},
  volume={53},
  pages={e2025GL117609},
  year={2026},
  publisher={Wiley Online Library},
  doi={10.1029/2025GL117609} 
}

@article{CharltonPerez2024,
  title={Do {AI} models produce better weather forecasts than physics-based models? {A} quantitative evaluation case study of {Storm Ciar{\'a}n}},
  author={Charlton-Perez, Andrew J and Dacre, Helen F and Driscoll, Simon and Gray, Suzanne L and Harvey, Ben and Harvey, Natalie J and Hunt, Kieran MR and Lee, Robert W and Swaminathan, Ranjini and Vandaele, Remy and others},
  journal={npj Climate and Atmospheric Science},
  volume={7},
  pages={93},
  year={2024},
  publisher={Nature Publishing Group UK London},
  doi={10.1038/s41612-024-00638-w}
}

@article{Pasche2025,
  title={Validating deep learning weather forecast models on recent high-impact extreme events},
  author={Pasche, Olivier C and Wider, Jonathan and Zhang, Zhongwei and Zscheischler, Jakob and Engelke, Sebastian},
  journal={Artificial Intelligence for the Earth Systems},
  volume={4},
  pages={e240033},
  year={2025},
  publisher={American Meteorological Society},
  doi={10.1175/AIES-D-24-0033.1}
}

@article{Rasp2024weatherbench2,
  title={{WeatherBench} 2: A benchmark for the next generation of data-driven global weather models},
  author={Rasp, Stephan and Hoyer, Stephan and Merose, Alexander and Langmore, Ian and Battaglia, Peter and Russell, Tyler and Sanchez-Gonzalez, Alvaro and Yang, Vivian and Carver, Rob and Agrawal, Shreya and others},
  journal={Journal of Advances in Modeling Earth Systems},
  volume={16},
  pages={e2023MS004019},
  year={2024},
  publisher={Wiley Online Library},
  doi={10.1029/2023MS004019}
}

@article{Lam2023graphcast,
  title={Learning skillful medium-range global weather forecasting},
  author={Lam, Remi and Sanchez-Gonzalez, Alvaro and Willson, Matthew and Wirnsberger, Peter and Fortunato, Meire and Alet, Ferran and Ravuri, Suman and Ewalds, Timo and Eaton-Rosen, Zach and Hu, Weihua and others},
  journal={Science},
  volume={382},
  pages={1416--1421},
  year={2023},
  publisher={American Association for the Advancement of Science},
  doi={10.1126/science.adi2336}
}

@article{Bi2023pangu,
  title={Accurate medium-range global weather forecasting with {3D} neural networks},
  author={Bi, Kaifeng and Xie, Lingxi and Zhang, Hengheng and Chen, Xin and Gu, Xiaotao and Tian, Qi},
  journal={Nature},
  volume={619},
  pages={533--538},
  year={2023},
  publisher={Nature Publishing Group UK London},
  doi={10.1038/s41586-023-06185-3}
}

@article{Bodnar2025aurora,
  title={A foundation model for the {Earth} system},
  author={Bodnar, Cristian and Bruinsma, Wessel P and Lucic, Ana and Stanley, Megan and Allen, Anna and Brandstetter, Johannes and Garvan, Patrick and Riechert, Maik and Weyn, Jonathan A and Dong, Haiyu and others},
  journal={Nature},
  volume={641},
  pages={1180--1187},
  year={2025},
  publisher={Nature Publishing Group UK London},
  doi={10.1038/s41586-025-09005-y}
}

@article{lang2024aifs,
  title={{AIFS}--{ECMWF}'s data-driven forecasting system},
  author={Lang, Simon and Alexe, Mihai and Chantry, Matthew and Dramsch, Jesper and Pinault, Florian and Raoult, Baudouin and Clare, Mariana CA and Lessig, Christian and Maier-Gerber, Michael and Magnusson, Linus and others},
  journal={arXiv preprint arXiv:2406.01465},
  year={2024}
}

@article{chen2023fuxi,
  title={{FuXi}: A cascade machine learning forecasting system for 15-day global weather forecast},
  author={Chen, Lei and Zhong, Xiaohui and Zhang, Feng and Cheng, Yuan and Xu, Yinghui and Qi, Yuan and Li, Hao},
  journal={npj Climate and Atmospheric Science},
  volume={6},
  pages={190},
  year={2023},
  publisher={Nature Publishing Group UK London},
  doi={10.1038/s41612-023-00512-1}
}

@article{Hersbach2020,
  title={The {ERA5} global reanalysis},
  author={Hersbach, Hans and Bell, Bill and Berrisford, Paul and Hirahara, Shoji and Hor{\'a}nyi, Andr{\'a}s and Mu{\~n}oz-Sabater, Joaqu{\'\i}n and Nicolas, Julien and Peubey, Carole and Radu, Raluca and Schepers, Dinand and others},
  journal={Quarterly Journal of the Royal Meteorological Society},
  volume={146},
  pages={1999--2049},
  year={2020},
  publisher={Wiley Online Library},
  doi={10.1002/qj.3803}
}

@article{Bonavita2024,
  title={On some limitations of current machine learning weather prediction models},
  author={Bonavita, Massimo},
  journal={Geophysical Research Letters},
  volume={51},
  pages={e2023GL107377},
  year={2024},
  publisher={Wiley Online Library},
  doi={10.1029/2023GL107377}
}

@article{Brenowitz2025,
  title={A practical probabilistic benchmark for {AI} weather models},
  author={Brenowitz, Noah D and Cohen, Yair and Pathak, Jaideep and Mahesh, Ankur and Bonev, Boris and Kurth, Thorsten and Durran, Dale R and Harrington, Peter and Pritchard, Michael S},
  journal={Geophysical Research Letters},
  volume={52},
  pages={e2024GL113656},
  year={2025},
  publisher={Wiley Online Library},
  doi={10.1029/2024GL113656}
}

@article{Loveday2025,
  title={Evaluating extreme precipitation forecasts: A threshold-weighted, spatial verification approach for comparing an {AI} weather prediction model against a high-resolution {NWP} model},
  author={Loveday, Nicholas and Hertneky, Tracy},
  journal={arXiv preprint arXiv:2510.25045},
  year={2025}
}

@article{Flach2019,
  title = {Performance Evaluation in Machine Learning: The Good,  the Bad,  the Ugly,  and the Way Forward},
  volume = {33},
  doi = {10.1609/aaai.v33i01.33019808},
  journal = {Proceedings of the AAAI Conference on Artificial Intelligence},
  publisher = {Association for the Advancement of Artificial Intelligence (AAAI)},
  author = {Flach,  Peter},
  year = {2019},
  pages = {9808--9814}
}

@inproceedings{
liao2021are,
title={Are We Learning Yet? {A} Meta Review of Evaluation Failures Across Machine Learning},
author={Thomas Liao and Rohan Taori and Inioluwa Deborah Raji and Ludwig Schmidt},
booktitle={Thirty-fifth Conference on Neural Information Processing Systems Datasets and Benchmarks Track},
year={2021},
url={https://openreview.net/forum?id=mPducS1MsEK}
}

@article{DimitriadisPuke2026,
  title={Statistical Inference for Score Decompositions},
  author={Dimitriadis, Timo and Puke, Marius},
  journal={arXiv preprint arXiv:2603.04275},
  year={2026}
}


\appendix

\section{Additional methodological details}
\label{app:method_details}

\subsection{Spatial aggregation}
\label{app:spatial_aggregation}
Let $\mathcal G$ denote the set of grid points in the evaluation domain. For a fixed model, variable, record type, lead time, and grid point $g \in \mathcal{G}$, let $Q_g$ denote a grid-point-level evaluation metric (e.g. an average score or decomposition component) based on $n_g$ evaluated cases. Let $w_g$ denote the grid-cell area weight computed from the latitude cell bounds: following \cite{Rasp2024weatherbench2}, if $\phi_g^-$ and $\phi_g^+$ are the lower and upper latitude bounds of the cell, then 
\[ 
    w_g \propto \sin(\phi_g^+)-\sin(\phi_g^-). 
\] 
We aggregate the grid-point-level quantities as 
\begin{equation} 
    Q_{\mathcal G} = \frac{\sum_{g\in\mathcal G} w_g n_g Q_g} {\sum_{g\in\mathcal G} w_g n_g} \label{eq:spatial_aggregation} 
\end{equation} 
where the normalising constant in $w_g$ cancels in the weighted average. For the owMSE, $n_g$ is the number of observed record-breaking events at grid point $g$. For scores evaluated over all forecast cases, $n_g=n$ is constant across grid points and therefore also cancels out. Grid points with $n_g=0$ do not contribute. The same aggregation is applied to the individual score decomposition components.

\subsection{Confidence intervals}
\label{app:confidence}


We quantify uncertainty in the spatially aggregated score $Q_{\mathcal G}$ in \eqref{eq:spatial_aggregation} using a normal approximation. At grid point $g \in \mathcal{G}$, let $Q_{g,j}$ denote a score for forecast--observation pair $j$ among the $n_g$ evaluated pairs, in which case $Q_g$ is the mean of these individual scores. Equation~\eqref{eq:spatial_aggregation} can then be written as
\[
    Q_{\mathcal G}
    =
    \frac{
        \displaystyle
        \sum_{g\in\mathcal G}w_g
        \sum_{j=1}^{n_g}Q_{g,j}
    }{
        \displaystyle
        \sum_{g\in\mathcal G}w_gn_g
    }.
\]

Define the total number of evaluated forecast--observation pairs across the
domain as $N=\sum_{g\in\mathcal G}n_g.$ Then, 
\[
    Q_{\mathcal G} = \frac{1}{N}
    \sum_{g\in\mathcal G}\sum_{j=1}^{n_g}a_gQ_{g,j},
\]
where
\[
    a_g
    =
    \frac{Nw_g}
         {\sum_{h\in\mathcal G}w_hn_h}.
\]
Since $Q_{\mathcal{G}}$ is a weighted sample mean, we can obtain confidence intervals by leveraging the central limit theorem. Following \citet{Zhang2026}, we estimate the variance of the normalized area-weighted individual scores $a_gQ_{g,j}$ by
\[
    \widehat{\sigma}^2
    =
    \frac{1}{N}
    \sum_{g\in\mathcal G}\sum_{j=1}^{n_g}
    \left(a_gQ_{g,j}-Q_{\mathcal G}\right)^2.
\]
For $\alpha \in (0, 1)$, a pointwise $(1-\alpha)$-confidence interval for $Q_{\mathcal G}$ is therefore
\begin{equation}
    \left[
    Q_{\mathcal G}
    -z_{1-\alpha/2}\sqrt{\frac{\widehat{\sigma}^2}{N}},
    \;
    Q_{\mathcal G}
    +z_{1-\alpha/2}\sqrt{\frac{\widehat{\sigma}^2}{N}}
    \right],
    \label{eq:confidence_interval}
\end{equation}
where $z_p$ denotes the $p$-quantile of the standard normal
distribution. We use $\alpha=0.01$, giving pointwise 99\% confidence
intervals. These intervals are constructed separately for each model,
variable, record type, and lead time. \citet{Zhang2026} additionally
apply the Delta-method to obtain confidence intervals for the root mean squared error (RMSE). No such transformation is required here since we report the mean scores directly.

\subsection{twMSE implementation details}
\label{app:twmse_details}
The record threshold $t_i$ varies across months and grid points. This poses no problem for computing the twMSE itself, but the empirical score decomposition in Section~\ref{sec:score_decompositions} requires a common scoring function. We therefore transform each forecast--observation pair to
\begin{equation}
    (\widetilde x_i,\widetilde y_i)
    =
    \begin{cases}
        (x_i-t_i,\,y_i-t_i), & \text{heat and wind records},\\
        (t_i-x_i,\,t_i-y_i), & \text{cold records}.
    \end{cases}
    \label{eq:tail_transform}
\end{equation}
On this scale, $\widetilde x_i>0$ indicates that a record-breaking event has been forecast, while $\widetilde y_i>0$ denotes that a record-breaking event has been observed.

We apply the fixed upper-tail score $s_{2,0}$ to all transformed pairs. For heat and wind records, we have
\[ 
    s_{2,0}(x_i-t_i,y_i-t_i) 
    = 
    s_{2,t_i}(x_i,y_i). 
\]
For cold records,
\begin{align*} 
    s_{2,0}(t_i-x_i,t_i-y_i) 
    &= (t_i-y_i)_+^2-(t_i-x_i)_+^2 -2\bigl[(t_i-y_i)-(t_i-x_i)\bigr](t_i-x_i)_+ \\ 
    &= (t_i-y_i)_+^2-(t_i-x_i)_+^2 -2(x_i-y_i)(t_i-x_i)_+ \\ 
    &= s_{1,t_i}(x_i,y_i). 
\end{align*}
Consequently, the twMSE can be computed for all three event types as
\[ 
    \operatorname{twMSE} = \frac{1}{n}\sum_{i=1}^n s_{2,0}(\widetilde x_i,\widetilde y_i). 
\]

The twMSE can be written as 
\begin{equation} 
    s_{2,0}(\widetilde x_i,\widetilde y_i) 
    = 
    \begin{cases} 
        (\widetilde x_i-\widetilde y_i)^2, & \widetilde x_i>0,\ \widetilde y_i>0 \quad\text{(hit)},\\ 
        \widetilde y_i^2, & \widetilde x_i\leq0,\ \widetilde y_i>0 \quad\text{(miss)},\\ 
        \widetilde x_i^2-2\widetilde x_i\widetilde y_i, & \widetilde x_i>0,\ \widetilde y_i\leq0 \quad\text{(false alarm)},\\ 
        0, & \widetilde x_i\leq0,\ \widetilde y_i\leq0 \quad\text{(correct non-detection)}. 
    \end{cases} 
    \label{eq:twmse_cases} 
\end{equation}

For the contribution shares reported in the results, let $\mathcal L$ denote the set of evaluated lead times and let $C_{g,i,\ell} \in \{\text{hit},\text{miss},\text{false alarm}\}$ denote the event category for case $i \in \{1, \dots, n\}$ at grid point $g \in \mathcal{G}$ and lead time $\ell\in\mathcal L$. Using the grid-cell area weights $w_g$ defined in Appendix~\ref{app:spatial_aggregation}, the pooled share attributable to category $c$ is
\begin{equation} 
    q_c = 
    \frac{ 
    \displaystyle \sum_{\ell\in\mathcal L}\sum_{g\in\mathcal G} w_g \sum_{i=1}^n s_{2,0}(\widetilde x_{g,i,\ell},\widetilde y_{g,i,\ell}) \mathbbm{1}\{C_{g,i,\ell}=c\} 
    }{ 
    \displaystyle \sum_{\ell\in\mathcal L}\sum_{g\in\mathcal G} w_g \sum_{i=1}^n s_{2,0}(\widetilde x_{g,i,\ell},\widetilde y_{g,i,\ell}) 
    }, \quad c\in\{\text{hit},\text{miss},\text{false alarm}\},
    \label{eq:twmse_shares} 
\end{equation}
where $\widetilde{x}_{g,i,\ell} \in \R$ is the $i$-th (transformed) prediction at grid point $g$ and lead time $\ell$, and $\widetilde{y}_{g,i,\ell}$ is the corresponding (transformed) observation. For lead-time-specific results, the sum over $\ell$ is omitted. The quantities $q_c$ are the proportions of the total twMSE attributable to hits, misses, and false alarms, and they therefore reflect both the frequency and magnitude of the respective errors. The three contributions sum to one whenever the denominator in Equation~\eqref{eq:twmse_shares} is positive.

The case-wise representation in Equation~\eqref{eq:twmse_cases} also reveals an asymmetry between misses and false alarms. 
Let $a>0$ and consider two cases that mirror each other at the threshold: a miss with $(\widetilde x_i,\widetilde y_i)=(-a,a)$ and a false alarm with $(\widetilde x_i,\widetilde y_i)=(a,-a)$. Both have the same unweighted squared error, $4a^2$, but their contributions to the twMSE are 
\begin{align*} 
    s_{2,0}(-a,a) &= a^2,\\ 
    s_{2,0}(a,-a) &= a^2-2a(-a)=3a^2. 
\end{align*} 
Thus, for these equally sized symmetric errors, the false alarm is penalised three times as much as the miss. This asymmetry is a property of the threshold-weighted squared error and should not be interpreted as an economic cost ratio. This likely occurs since the twMSE remains consistent for the mean functional, meaning it is designed to encourage predictions of the mean of the outcome, rather than predictions in the tail.

\subsection{twPCRPS implementation details}
\label{app:twpcrps_details}

The optimality result of \citet[Theorem 2]{Henzi2021idr} concerns the average of a common scoring rule across all fitted cases. It does not generally guarantee optimality when IDR is fitted to the original forecast--observation pairs and the resulting distributions are evaluated using thresholds that vary between cases. We therefore apply the transformation in Equation~\eqref{eq:tail_transform} before fitting IDR, separately for each record type, and evaluate the resulting distributions using the upper-tail twCRPS at threshold zero, as in Equation~\eqref{eq:twpcrps}. This retains in-sample optimality subject to the stochastic ordering constraint on the transformed forecasts.

The ordering constraint now associates larger transformed forecasts $\widetilde x_i$ with stochastically larger predictive distributions for $\widetilde y_i$. It therefore concerns forecast and outcome departures from their monthly record thresholds in the direction of the relevant extreme. The rationale is that larger forecast departures should indicate a greater chance of exceeding the relevant record by any given amount, after accounting for seasonal differences in record levels. When the threshold is constant across all cases within a fit, transforming and refitting yields the same twPCRPS as fitting the original pairs and evaluating the corresponding tail at the original threshold. With varying thresholds, the transformation changes the regression problem and can change the ordering of the forecasts, so the resulting twPCRPS generally differs. This contrasts with the twMSE, whose individual score values are preserved by the transformation.

\subsection{Computational resources}
\label{app:computational_resources}
We evaluate existing WeatherBench 2 forecasts and neither train nor run the forecasting models. All final $0.25^\circ$ experiments were performed without GPUs on a shared internal server equipped with two Intel Xeon E5-2680 v3 processors at 2.50 GHz, providing 24 physical cores and 48 hardware threads, and 236 GB of memory. The MSE and owMSE computations used Dask's default threaded scheduler, the twMSE computations used 24 Dask worker processes, and the PCRPS and twPCRPS computations used the Apache Beam DirectRunner with 40 worker processes. CPU cores were not reserved for the MSE, owMSE, and twMSE computations, whereas the PCRPS and twPCRPS computations used 40 reserved hardware threads, corresponding to 20 physical cores. The observed wall-clock times reflect different resource allocations and parallelisation settings and should therefore not be interpreted as controlled comparisons of computational complexity.
The MSE, owMSE, and twMSE computations, including their score decompositions, required approximately 9.7, 5.2, and 2.2 hours, respectively, while the PCRPS and twPCRPS computations required 169.5 hours. Together, the final $0.25^\circ$ experiments required approximately 186.6 wall-clock hours.

The forecast and verification inputs occupied approximately 81 GB, while the stored outputs occupied approximately 118 GB. Of these outputs, 116 GB comprised PCRPS and twPCRPS scores stored separately for every initialisation. This simplified the implementation and avoided retaining the scores in memory, but required substantial storage. Only their temporal means, variances, and counts are needed for the reported spatial aggregates and confidence intervals, so a more storage-efficient implementation could substantially reduce this footprint. Additional lower-resolution exploratory and debugging runs required further compute but are not included in these estimates.

The computations are performed separately for each combination of model, variable, grid point, and lead time. The total computational cost therefore grows linearly with the number of such combinations, which can be processed in parallel. For $n$ evaluated forecast--observation pairs within one combination, calculating MSE, owMSE, and twMSE requires $\mathcal O(n)$ operations. The score decompositions additionally require sorting the forecasts and scale as $\mathcal O(n\log n)$. EasyUQ recalibration has an $\mathcal O(n^2)$ upper bound \citep[Section~3.4]{Walz2024easyuq}, while evaluating the fitted predictive distributions for PCRPS and twPCRPS can require up to $\mathcal O(n^2)$ operations \citep{Gneiting2026pcrps}.

\section{Score decompositions}
\label{app:score_decomp}

\subsection{MSE}

Suppose that the prediction is a random variable, denoted by $X \in \R$, and we wish to evaluate our prediction using the expected score $\E[s(X, Y)]$. Like the average score in Section \ref{sec:eval_and_decomp}, the expected score can be decomposed into terms that quantify the forecast uncertainty, forecast information content, and forecast miscalibration, allowing practitioners to study the contributions of these components to the total score. If $s$ is consistent for the mean functional, we can decompose $\E[s(X, Y)]$ as
\begin{equation}\label{eq:decomp_pop}
{\small 
\E[s(X, Y)] = \underbrace{\E[s(\E[Y], Y)]}_{\mathrm{UNC}} - \underbrace{\left\{\E[s(\E[Y], Y)] - \E[s(\E[Y|X], Y)]\right\}}_{\mathrm{DSC}} + \underbrace{\left\{\E[s(X, Y)] - \E[s(\E[Y|X], Y)]\right\}}_{\mathrm{MCB}}.
}
\end{equation}
For the squared error, this simplifies to
\begin{equation}\label{eq:decomp_se}
    \E[(X - Y)^2] = \underbrace{\E[(\E[Y] - Y)^2]}_{\mathrm{UNC}} - \underbrace{\E[(\E[Y] - \E[Y \mid X])^2]}_{\mathrm{DSC}} + \underbrace{\E[(X - \E[Y \mid X])^2]}_{\mathrm{MCB}}.
\end{equation}
The interpretation of the components is analogous to the interpretation of the sample components in Section \ref{sec:eval_and_decomp}; the sample components can be interpreted as finite sample estimates of the population decomposition terms in Equation~\eqref{eq:decomp_pop}. In particular, it becomes clear from Equation~\eqref{eq:decomp_se} that all terms of the decomposition are non-negative, with the discrimination term equal to zero if and only if $\E[Y] = \E[Y \mid X]$ almost surely, while the miscalibration term is equal to zero if and only if $X = \E[Y \mid X]$ almost surely. This interpretation holds more generally for any scoring function $s$ that is strictly consistent for the mean.

In practice, we evaluate predictions using prediction-observation pairs $(x_1, y_1), \dots, (x_n, y_n)$, which can be interpreted as realisations of the random pair $(X, Y)$. The expected score can then be estimated using the sample mean score over the $n$ pairs, the expectation of $Y$ can be estimated using the sample mean $\bar{y} = \sum_{i=1}^n y_i / n$, and the conditional expectation $\E[Y \mid X]$ can be estimated by recalibrating the original predictions $x_1, \dots, x_n$. That is, the recalibrated values $x_1^{(r)}, \dots, x_n^{(r)}$ can be interpreted as estimates of $\E[Y \mid X = x_i]$. This decomposition is presented in Equation~\eqref{eq:decomp_samp}. In this sample decomposition, the miscalibration term $\mathrm{\widehat{MCB}}$ quantifies the distance between $x_i$ and $x_i^{(r)}$. The condition $x_i^{(r)} = x_i$ can be interpreted as a sample-level analog of the standard calibration condition of mean forecasts, $\E[Y \mid X] = X$, highlighting that $\mathrm{\widehat{MCB}}$ estimates how much $X$ differs from $\E[Y \mid X]$, providing a measure of miscalibration for mean predictions.


\subsection{Threshold-weighted MSE}\label{app:score_decomp_twmse}


Since the threshold-weighted squared error of \cite{Taggart2022} is also consistent for the mean functional, we can analogously decompose the expected score into terms quantifying forecast uncertainty, discrimination, and miscalibration. We focus on the score $s_{2,t}$ introduced in Section \ref{sec:eval_and_decomp:subsec:eval}, where we are interested in threshold exceedances, though an analogous decomposition can be applied to $s_{1,t}$. The resulting decomposition terms are a function of the threshold $t$,
\begin{align}\label{eq:decomp_wse}
\begin{split}
  \E[s_{2,t}(X, Y)] &= \underbrace{\E[s_{2,t}(\E[Y], Y)]}_{\mathrm{UNC}_t} - \underbrace{\left\{\E[s_{2,t}(\E[Y], Y)] - \E[s_{2,t}(\E[Y|X], Y)]\right\}}_{\mathrm{DSC}_t} \\
  &\quad + \underbrace{\left\{\E[s_{2,t}(X, Y)] - \E[s_{2,t}(\E[Y|X], Y)]\right\}}_{\mathrm{MCB}_t}.
\end{split}
\end{align}

The interpretation of the decomposition terms is analogous to before, except the focus is now explicitly on the region above the threshold. Since the twMSE focuses only on threshold exceedances, the discrimination ($\mathrm{DSC}_t$) and miscalibration ($\mathrm{MCB}_t$) terms of this decomposition now quantify how informative and how miscalibrated the predictions are when interest is on extreme events. This is formalised in the following proposition.

\begin{prop}\label{prop:decomp_wse}
    Let $\mathrm{DSC}_t$ and $\mathrm{MCB}_t$ be defined as in Equation~\eqref{eq:decomp_wse}. Then,
    \begin{enumerate}
        \item[i)] $\mathrm{DSC}_t \ge 0$ with equality if and only if $(\E[Y \mid X] - t)_+ = (\E[Y] - t)_+$ almost surely.
        \item[ii)] $\mathrm{MCB}_t \ge 0$ with equality if and only if $(\E[Y \mid X] - t)_+ = (X - t)_+$ almost surely.
    \end{enumerate}
\end{prop}
\bigskip
\begin{remark}
    The condition for equality in part $i)$ can alternatively be expressed as $\E[Y \mid X = x] > t$ implies $\E[Y \mid X = x] = \E[Y]$, and $\E[Y \mid X = x] \le t$ implies $\E[Y] \le t$. Similarly, the condition for equality in part $ii)$ requires that $\E[Y \mid X = x] > t$ implies $\E[Y \mid X = x] = x$, and $\E[Y \mid X = x] \le t$ implies $x \le t$.
\end{remark}

\begin{proof}[Proof of Proposition \ref{prop:decomp_wse}]
    Consider the expected score $\E[s_{2,t}(X, Y)]$. Using the Tower Property of conditional expectations, this can be rewritten as
    \[
    \E[s_{2,t}(X, Y)] = \E[\E[s_{2,t}(X, Y) \mid X]].
    \]
    Recall that the expected score $\E[s_{2,t}(x, Y)]$ is minimised when $x = \E[Y]$ when $\E[Y] > t$, and by any $x \le t$ when $\E[Y] \le t$. This is equivalent to saying $(x - t)_+ = (\E[Y] - t)_+$. Hence, $\E[s_{2,t}(X, Y) \mid X]$ is minimised when $(X - t)_+ = (\E[Y \mid X] - t)_+$ almost surely. It follows that $\E[Y \mid X]$ receives the optimal expected score (conditional on $X$), though not uniquely. Thus, we must have that 
    \[
    \E[s_{2,t}(X, Y)\mid X] - \E[s_{2,t}(\E[Y \mid X], Y) \mid X] \ge 0,
    \]
    with equality if and only if $(X - t)_+ = (\E[Y \mid X] - t)_+$. By taking expectations, we have that 
    \[
    \mathrm{MCB}_t = \E[s_{2,t}(X, Y)] - \E[s_{2,t}(\E[Y \mid X], Y)] \ge 0, 
    \]
    for any threshold $t$. Since this is an expectation over the conditional score differences, which are guaranteed to be non-negative, equality holds under the same conditions.
    
    An analogous result holds for the discrimination term, where we have that
    \[
    \E[s_{2,t}(\E[Y], Y)\mid X] - \E[s_{2,t}(\E[Y \mid X], Y) \mid X] \ge 0, 
    \]
    with equality if and only if $(\E[Y \mid X] - t)_+ = (\E[Y] - t)_+$. The discrimination term is the expectation of this, and is therefore similarly guaranteed to be non-negative and equal to zero under the same conditions.
\end{proof}

Proposition \ref{prop:decomp_wse} shows that the miscalibration term in Equation~\eqref{eq:decomp_wse} assesses whether $(\E[Y \mid X] - t)_+ = (X - t)_+$ almost surely. We refer to this condition as \emph{local calibration} of mean forecasts. This is weaker than the standard calibration condition $X = \E[Y \mid X]$ since it assesses calibration only above a threshold; this is similar in essence to the definition of tail auto-calibration of probabilistic forecasts in \cite{AllenEtAl2025}. In particular, local calibration only requires that $X = \E[Y \mid X]$ when the conditional expectation of $Y$ given $X$ is extreme. This highlights how the decomposition focuses only on calibration when forecasting extreme outcomes. 

A lower threshold corresponds to a stronger notion of calibration, and standard calibration can be interpreted as local calibration in the limit $t \rightarrow -\infty$. However, the concept of local calibration is still useful, since it allows us to assess whether predictions are calibrated when forecasting extreme outcomes. For example, it may be the case that the predictions are not calibrated, but they do issue calibrated forecasts for extremes. Moreover, from the measure of miscalibration assessed by decompositions of the unweighted squared error, it is not possible to determine how much of this miscalibration is attributable to predictions of extremes. However, this notion of local calibration is specific to mean forecasts, and it thus may not be meaningful when predicting other functionals.

We can similarly introduce sample estimates of the decomposition terms for the weighted squared error by replacing $s$ with $s_{2,t}$ in Equation~\eqref{eq:decomp_samp}. Since the threshold-weighted squared error is consistent for the mean functional, it can be expressed as a linear combination of elementary scoring functions \citep{Ehm2016}. \citet[Theorem 3.1]{GneitingResin2023} demonstrate that the PAV algorithm can be used to generate optimal predictions with respect to any scoring function in this form. Hence, since the sample mean $\bar{y}$ and the original predictions $x_i$ are possible solutions to the PAV algorithm, it follows that the sample estimates of the discrimination and miscalibration terms are themselves non-negative. We additionally have that the discrimination term is equal to zero if and only if $(x_i^{(r)} - t)_+ = (\bar{y} - t)_+$, and the miscalibration term is equal to zero if and only if $(x_i - t)_+ = (x_i^{(r)} - t)_+$. These correspond to sample versions of the requirements in Proposition \ref{prop:decomp_wse}.

\section{Additional results}
\label{app:additional_results}
Figure~\ref{fig:twmse_contributions_0p25_land} separates the twMSE contributions from hits, misses, and false alarms in Table~\ref{tab:record_contingency_land} by lead time. The absolute contributions to the twMSE are shown, rather than the percentage contributions. The broad contrast between physics-based and AI models persists throughout the forecast range. The HRES twMSE is dominated by false alarms, whose contribution increases strongly with lead time. For the AI models, misses generally make up a larger part of the score, especially for Pangu-Weather and Aurora, while GraphCast exhibits a more mixed composition.

\begin{figure}
  \centering
  \includegraphics[
    width=\textwidth
  ]{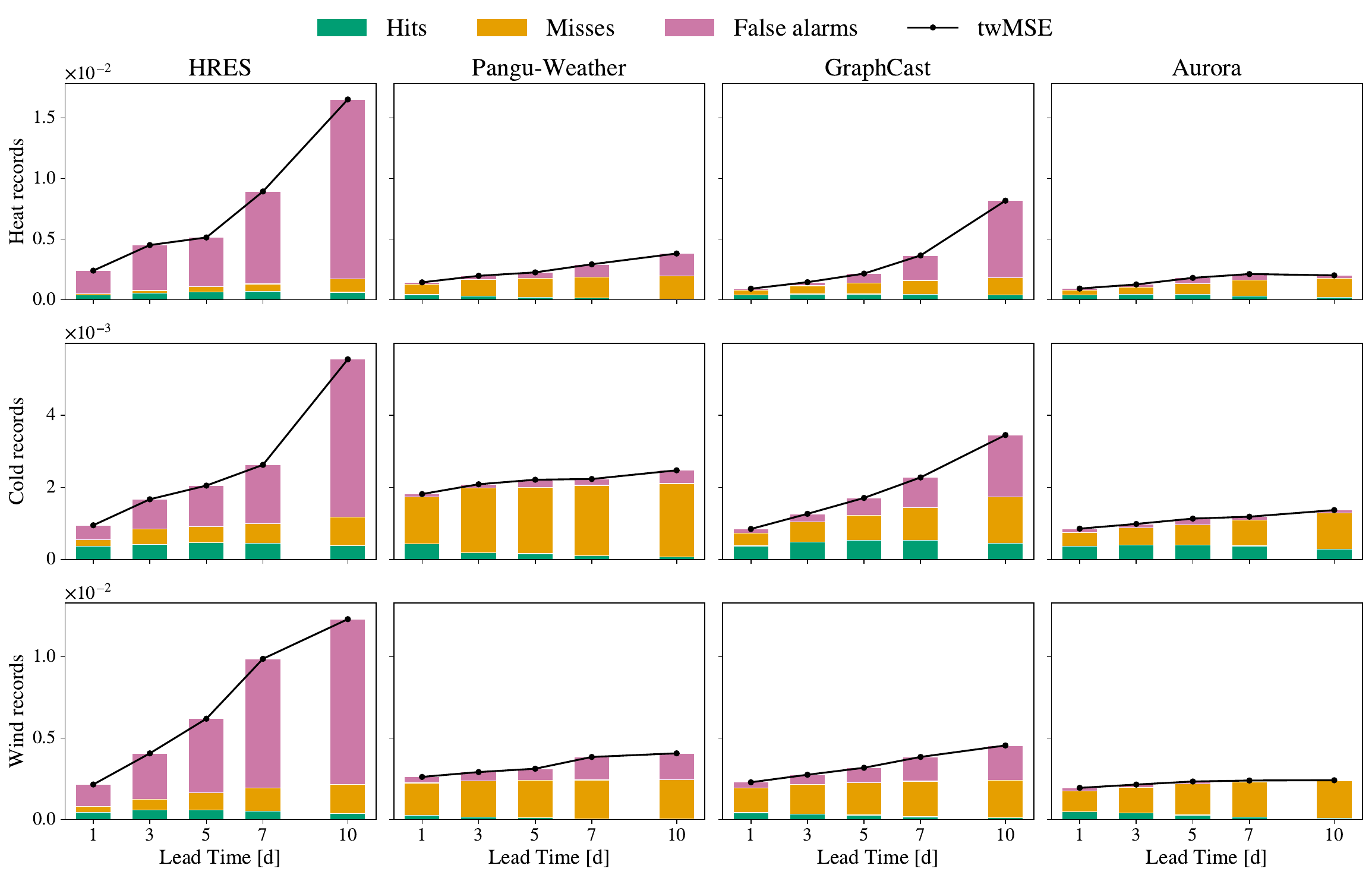}
  \caption{
    Lead-time-specific contributions of hits, misses, and false alarms to the average twMSE for forecasts of record-breaking events. Stacked bars show the component contributions, while the black line shows their sum, equal to the total average twMSE. Correct non-detections contribute zero. The scale of the $y$-axis is shared across models within each record type, but differs between record types.
  }
  \label{fig:twmse_contributions_0p25_land}
\end{figure}

Figure~\ref{fig:owmse_decomposition_0p25_land} complements the twMSE decomposition in Figure~\ref{fig:twmse_decomposition_0p25_land} by showing the corresponding components when decomposing the owMSE. In this case, calibration corresponds to $X = \E[Y \mid X, Y > t]$, which is not necessarily a property that we want our forecasts to satisfy, since it requires that an extreme event is always predicted. The miscalibration component again differs more strongly between models than the discrimination component. Since the physics-based model predicts extreme events more frequently, it generally receives the smallest miscalibration component, while Pangu-Weather has the largest. This contrasts sharply with the twMSE decomposition, under which HRES has the largest and Aurora the smallest miscalibration at longer lead times. The discrimination components exhibit much less variation between models: HRES generally has the greatest discrimination at short lead times, whereas Aurora has the greatest at longer lead times.

\begin{figure}
  \centering
  \includegraphics[width=\linewidth]{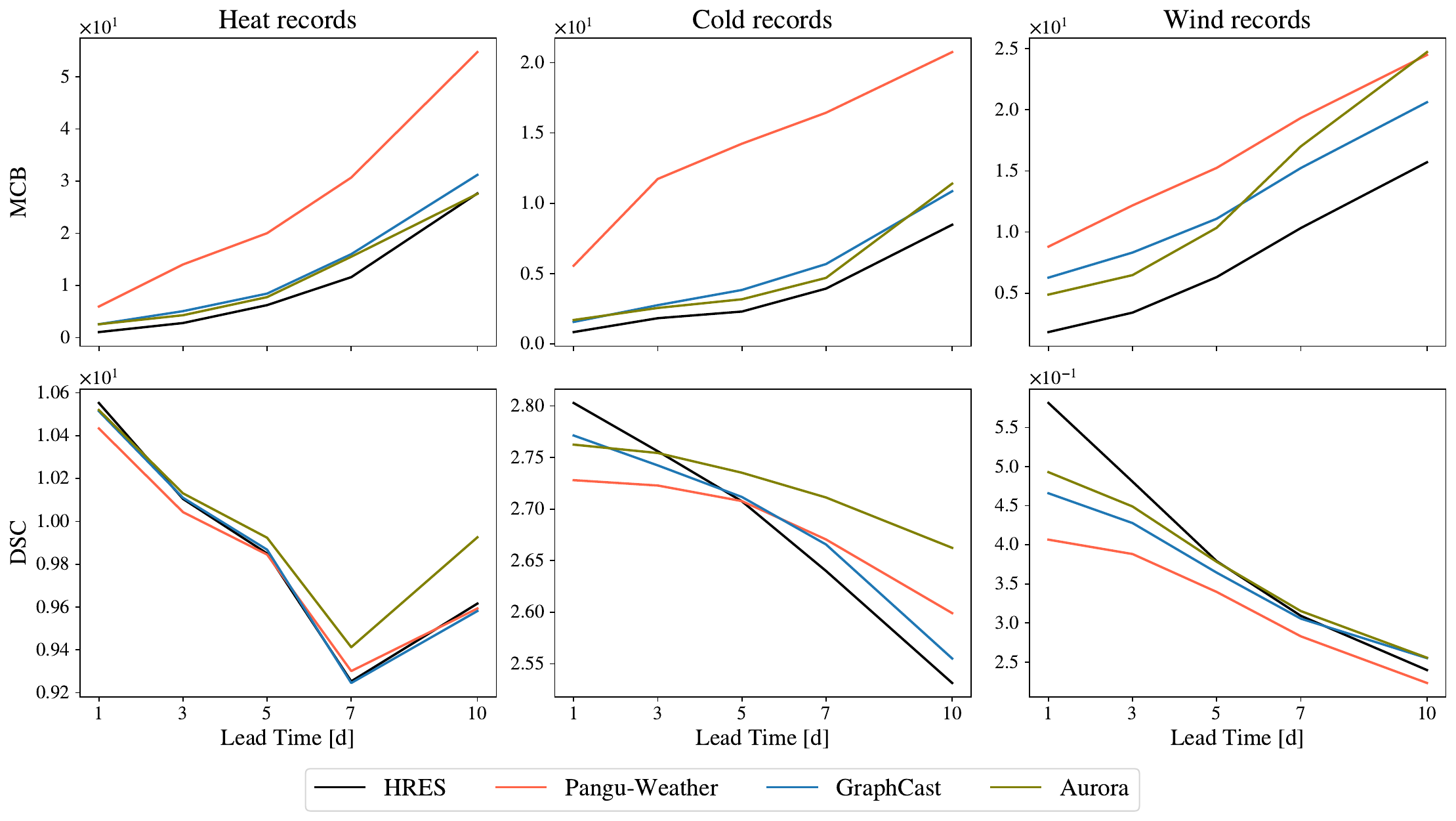}
  \caption{
    Miscalibration (MCB) and discrimination (DSC) components of the owMSE decomposition for forecasts of heat, cold, and wind record-breaking events. Grid point components are weighted by the number of observed record-breaking events. Note the different scales on the $y$-axes.
  }
  \label{fig:owmse_decomposition_0p25_land}
\end{figure}

Figure~\ref{fig:relative_miscalibration_0p25_land} expresses the miscalibration of each model relative to its corresponding score. This can be interpreted as the fraction of the total score that is removed by (in-sample) recalibration. For the owMSE, this fraction is at least $0.85$ for every model, record type, and lead time, and reaches approximately $0.95$ by day 10. HRES consistently has the lowest relative miscalibration, but recalibration substantially reduces the owMSE for all four models. For the twMSE, the relative miscalibration is generally lower and more model-dependent. Miscalibration generally contributes most to the average score for the physics-based model. 
This suggests that miscalibration dominates the owMSE much more than the twMSE. This can perhaps be explained by the fact that the miscalibration term that the owMSE assesses is not very meaningful, so that forecasts are much closer to being calibrated according to the twMSE requirement in Proposition \ref{prop:decomp_wse}.

\begin{figure}
  \centering
  \includegraphics[width=\linewidth]{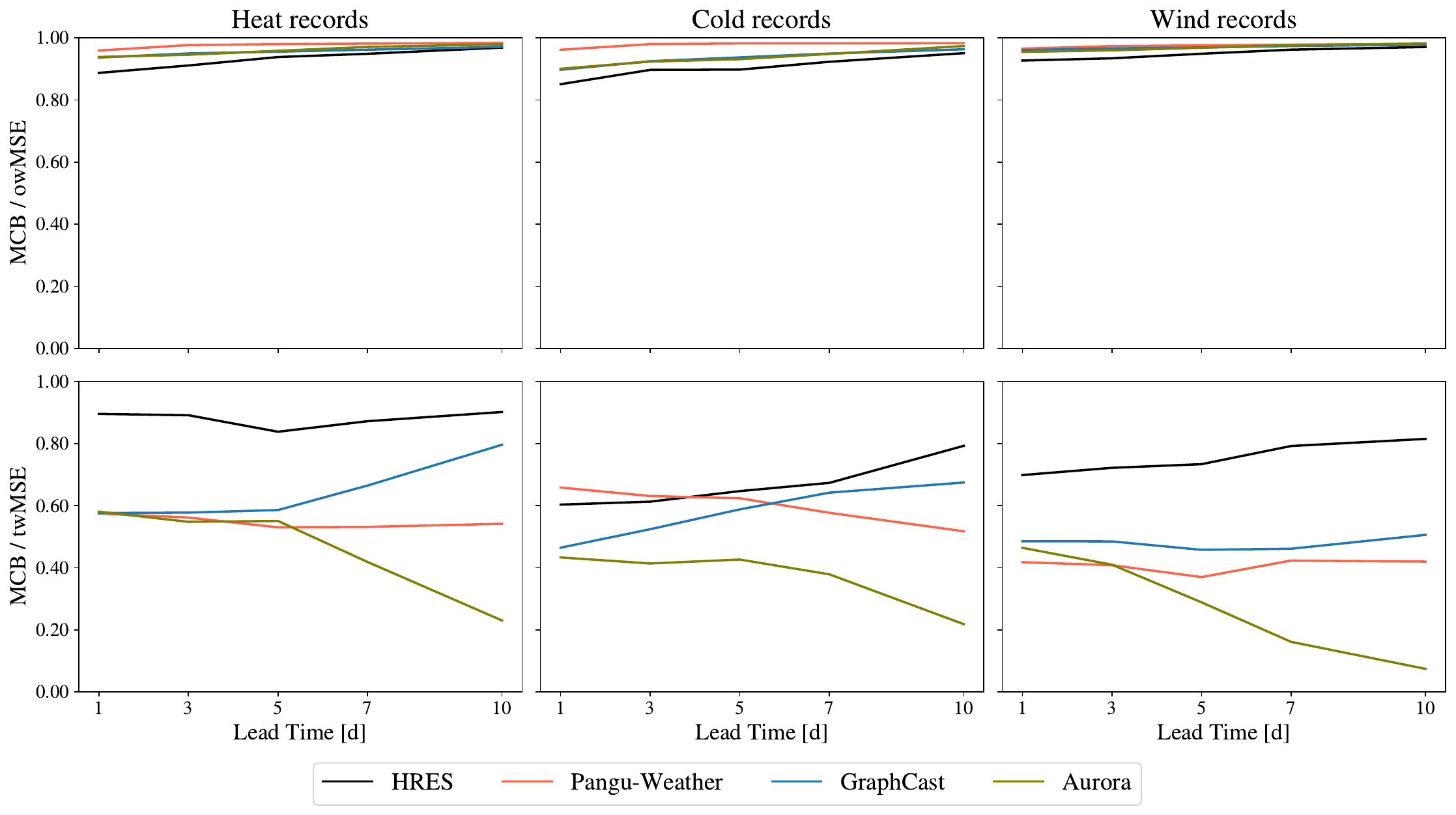}
  \caption{
    Contribution of forecast miscalibration to the total average scores. Results are shown for the owMSE (top row) and the twMSE (bottom row). The $y$-axis represents the fraction of the total score that is removed by in-sample isotonic recalibration. Values close to one therefore indicate that most of the score is attributable to miscalibration.
  }
  \label{fig:relative_miscalibration_0p25_land}
\end{figure}

\subsection{Global results}\label{app:global}

The results in the main text are aggregated across all grid points over land, excluding Antarctica. This is done since we are particularly interested in the impacts associated with extreme events, and impacts are largest over land. Nonetheless, to check the robustness of these results, we present results here for scores aggregated across all grid points (including ocean grid points, and including Antarctica).

Figure~\ref{fig:main_scores_0p25_global} presents the global counterparts of the land-based scores shown in Figure~\ref{fig:main_scores_0p25_land}. Over all forecast cases, the model rankings remain broadly unchanged: Aurora issues the most accurate temperature forecasts at medium and long lead times, while GraphCast receives the lowest MSE when forecasting wind speed at long lead times, and Aurora wind speed forecasts generally perform best under the PCRPS. The rankings under the owMSE are more sensitive to the spatial domain. The physics-based model retains the lowest owMSE for cold and wind records, but Aurora outperforms it for global heat records across all lead times. Under the twPCRPS, Aurora and GraphCast generally perform best, although HRES receives the lowest scores for cold and wind records at short lead times. Differences between models narrow at longer lead times, particularly for cold and wind records. Aurora shows the lowest twMSE for all three record types, whereas HRES shows the strongest increase in twMSE and ultimately has the highest values.

\begin{figure}
  \centering
  \includegraphics[width=\linewidth]{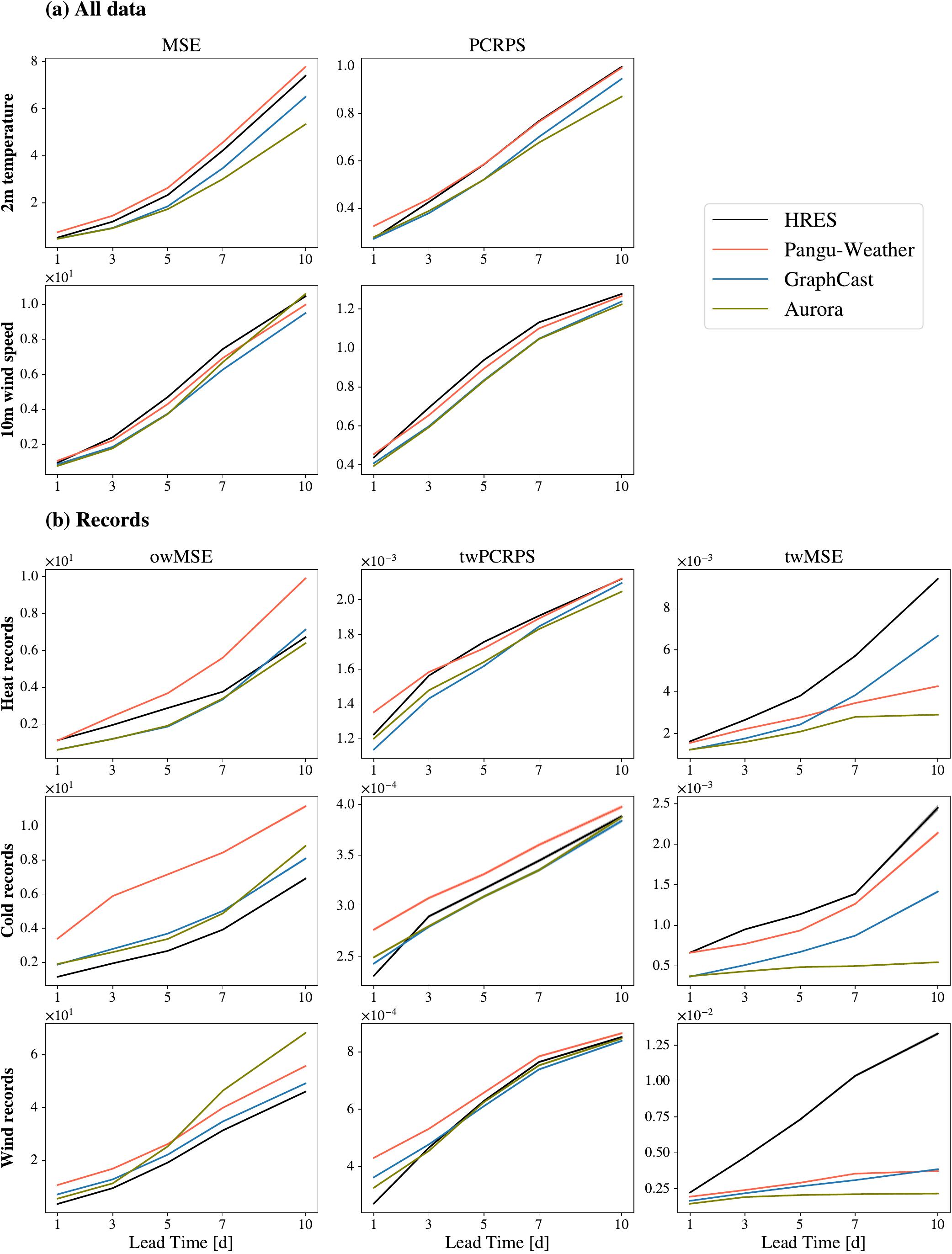}
  \caption{
    Average global scores for the four models.
    Panel~(a) shows MSE and PCRPS over all forecast cases. Panel~(b) shows owMSE, twPCRPS, and twMSE for heat, cold, and wind records. Shaded regions denote pointwise 99\% confidence intervals. Lower values indicate better performance. MSE, owMSE, and twMSE are reported in $\mathrm{K}^2$ for temperature and $\mathrm{m}^2\,\mathrm{s}^{-2}$ for wind speed; PCRPS and twPCRPS are reported in $\mathrm{K}$ and $\mathrm{m}\,\mathrm{s}^{-1}$, respectively.
  }
  \label{fig:main_scores_0p25_global}
\end{figure}

The global twMSE decomposition in Figure~\ref{fig:twmse_decomposition_0p25_global} closely resembles the land-based results in Figure~\ref{fig:twmse_decomposition_0p25_land}. The miscalibration of the physics-based model increases strongly with lead time, and is the largest for all three record types at every lead time. Aurora appears the most calibrated at longer lead times. The discrimination component decreases for every model and record type, with smaller differences between models. HRES has the largest discrimination at day 1 for cold and wind records, while Aurora has the largest discrimination at longer lead times for heat and cold records, and GraphCast for wind records. Thus, as over land, the global twMSE rankings are driven primarily by differences in the miscalibration component rather than discrimination.

\begin{figure}
  \centering
  \includegraphics[width=\linewidth]{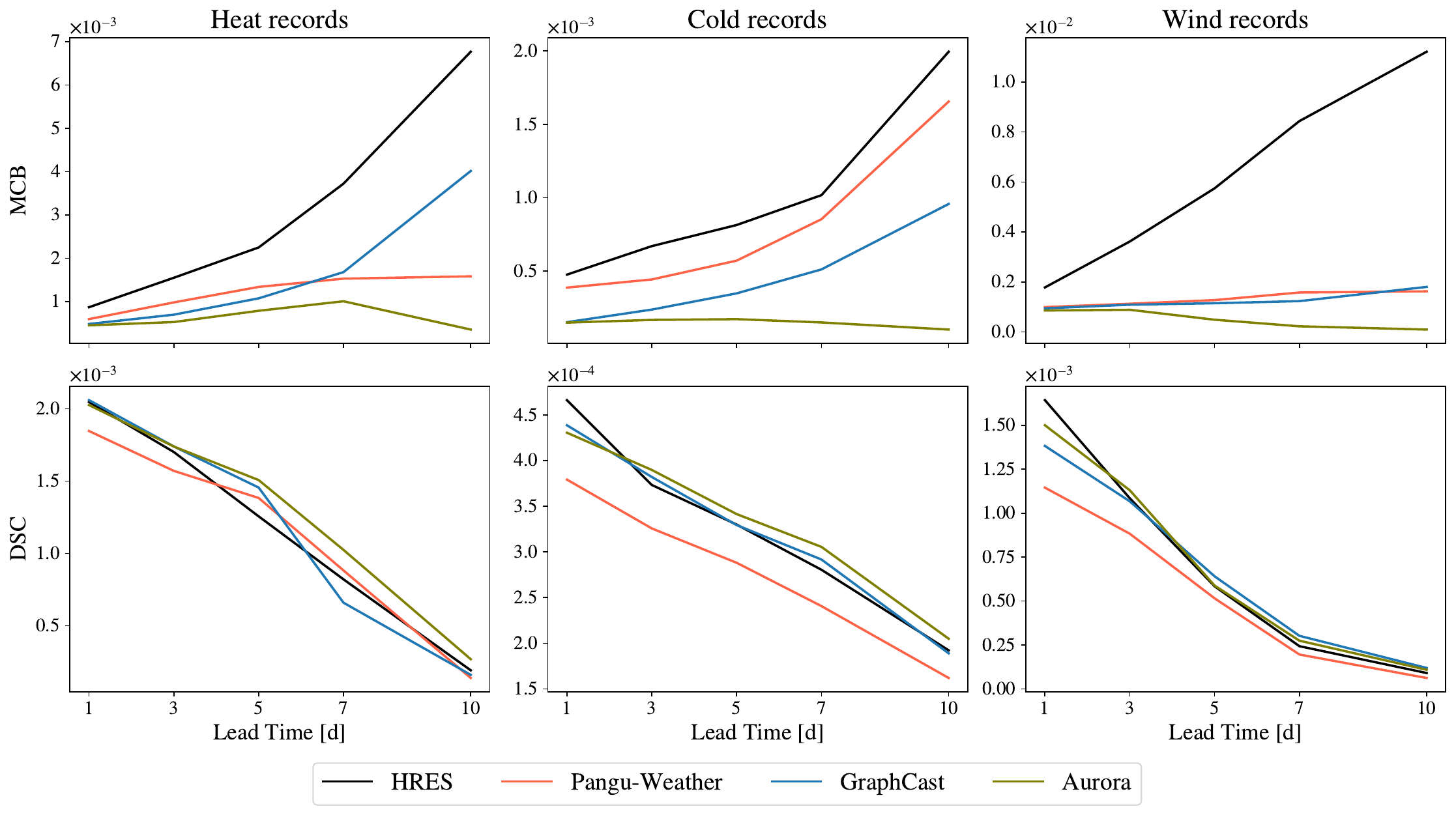}
  \caption{
    Miscalibration (MCB) and discrimination (DSC) components of the decomposition of the average global twMSE when evaluating predictions for heat, cold, and wind records. Lower MCB indicates better calibrated predictions, higher DSC indicates more informative predictions. Note the different scales on the $y$-axes.
  }
  \label{fig:twmse_decomposition_0p25_global}
\end{figure}

Table~\ref{tab:record_contingency_global} reports the global counterparts of the land-based contingency table statistics and twMSE contributions in Table~\ref{tab:record_contingency_land}. The broad twMSE contribution pattern is robust to the choice of domain. False alarms constitute the largest contribution to the HRES score for heat, cold, and wind records, accounting for $60.3\%$, $71.5\%$, and $78.8\%$, respectively. In contrast, misses form the largest contribution for every AI model and record type.

The contingency table statistics nevertheless reveal substantial differences between the land and global domains. In particular, the HRES frequency bias for heat records decreases from $1.521$ over land to $0.364$ globally, reversing from overforecasting over land to underforecasting globally. Despite these domain-dependent differences in record-forecast frequency, HRES scores remain predominantly false-alarm-driven, whereas AI model scores remain dominated by misses.

\begin{table}
  \caption{
    Global miss rate (MR), false alarm ratio (FAR), and frequency bias (FB) for the four models and three types of record-breaking events. The percentage contributions of hits, misses, and false alarms to global twMSE are also shown, with the largest contribution in each row in bold. Results are pooled over all initialisation times, grid points, and lead times. The twMSE contributions are latitude weighted.
  }
  \label{tab:record_contingency_global}
  \centering
  \begin{tabular}{@{}llrrr@{\hspace{14pt}}rrr@{}}
    \toprule
    &
    & \multicolumn{3}{c}{Contingency statistics}
    & \multicolumn{3}{c}{twMSE share [\%]} \\
    \cmidrule(lr){3-5}
    \cmidrule(lr){6-8}
    Record type
    & Model
    & MR [\%]
    & FAR [\%]
    & FB
    & Hit
    & Miss
    & False alarm \\
    \midrule

    \multirow{4}{*}{Heat}
    & HRES          & 83.87 & 55.64 & 0.364 &  9.41 & 30.30 & \textbf{60.29} \\
    & Pangu-Weather & 90.39 & 46.26 & 0.179 & 10.84 & \textbf{61.79} & 27.37 \\
    & GraphCast     & 84.95 & 49.36 & 0.297 & 11.88 & \textbf{50.36} & 37.76 \\
    & Aurora        & 88.38 & 28.22 & 0.162 & 15.25 & \textbf{77.11} &  7.63 \\

    \addlinespace
    \multirow{4}{*}{Cold}
    & HRES          & 69.33 & 65.25 & 0.882 & 11.17 & 17.31 & \textbf{71.51} \\
    & Pangu-Weather & 93.07 & 72.58 & 0.253 &  5.86 & \textbf{55.17} & 38.97 \\
    & GraphCast     & 84.99 & 48.46 & 0.291 & 19.13 & \textbf{44.64} & 36.22 \\
    & Aurora        & 84.06 & 29.18 & 0.225 & 24.55 & \textbf{66.72} &  8.74 \\

    \addlinespace
    \multirow{4}{*}{Wind}
    & HRES          & 47.26 & 45.33 & 0.965 &  5.15 & 16.04 & \textbf{78.81} \\
    & Pangu-Weather & 67.39 & 32.51 & 0.483 &  5.41 & \textbf{64.25} & 30.34 \\
    & GraphCast     & 81.39 & 37.79 & 0.299 &  9.05 & \textbf{62.02} & 28.93 \\
    & Aurora        & 62.97 & 18.68 & 0.455 & 12.02 & \textbf{82.00} &  5.99 \\

    \bottomrule
  \end{tabular}
\end{table}



\end{document}